\newif\ifconfver
\confverfalse     
\confvertrue        

\newif\ifplainver  
 \plainverfalse

\ifplainver
    \confverfalse   
\fi

\ifconfver
     \documentclass[10pt,twocolumn,twoside]{IEEEtran}
\else
    \ifplainver
        \documentclass[11pt]{article}
        \usepackage{fullpage}
    \else
        \documentclass[11pt,draftcls,onecolumn]{IEEEtran}
    \fi
\fi

\makeatletter
\newcommand{\settitle}{\@maketitle}
\makeatother

\usepackage{subfigure}
\usepackage[utf8]{inputenc}
\usepackage{bm, amsfonts, amsmath, amsthm, shortcuts_OPT}
\usepackage{bbm}
\usepackage{algorithm,algorithmic}
\usepackage{mathrsfs, enumitem}
\usepackage{tikz-network}
\usepackage{pstricks}
\usetikzlibrary{shapes}
\usepackage{pgfplots}
\usetikzlibrary{matrix}
\usetikzlibrary{calc}
\DeclareUnicodeCharacter{2212}{−}
\usepgfplotslibrary{groupplots,dateplot}
\usetikzlibrary{patterns,shapes.arrows}
\pgfplotsset{compat=newest}
\usepackage[colorlinks=true,
linkcolor=green!25!black,
urlcolor=blue,
citecolor=blue]{hyperref}       
\usepackage{cleveref}
\usepackage{graphicx} 
\usepackage{epstopdf}
\usepackage{booktabs}
\usepackage{colortbl}
\usepackage{multirow}

\newtheorem{Lemma}{Lemma}
\newtheorem{Prop}{Proposition}
\newtheorem{Prob}{Problem}
\newtheorem{Theorem}{Theorem}
\newtheorem{Def}{Definition}

\newtheorem{assumption}{H\!\!}
\newtheorem{Remark}{Remark}

\usepackage{cite}
\usepackage{tikz}

\usepackage{stackengine}
\usepgfplotslibrary{fillbetween}
\usepackage{stackengine}
\usepackage{stfloats}
\usepackage[font=normal]{caption}

\allowdisplaybreaks

\title{Detecting Low Pass Graph Signals via Spectral Pattern: Sampling Complexity and Applications}
\author{Chenyue Zhang, Yiran He, Hoi-To Wai\thanks{A preliminary version of this work has been presented at ICASSP 2021 \cite{he2021identifying}. 
The authors are with the Department of SEEM, The Chinese University of Hong Kong, Shatin, Hong Kong SAR of China. E-mails: \url{czhang@se.cuhk.edu.hk}, \url{yrhe@se.cuhk.edu.hk}, \url{htwai@se.cuhk.edu.hk}. This work is supported in part by CUHK Direct Grant \#4055135 and HKRGC Project \#24203520.}}
\date{\today}
                     
\begin{document}

{\let\newpage\relax\maketitle}


\begin{abstract}
This paper proposes a blind detection problem for low pass graph signals. Without assuming knowledge of the exact graph topology, we aim to detect if a set of graph signal observations are generated from a low pass graph filter. Our problem is motivated by the widely adopted assumption of low pass (a.k.a.~smooth) signals required by existing works in graph signal processing (GSP), as well as the longstanding problem of network dynamics identification. Focusing on detecting low pass graph signals on modular graphs whose cutoff frequency coincides with the number of clusters in the graph, we study and leverage the unique spectral pattern exhibited by such low pass graph signals to devise two detectors: one is based on Perron-Frobenius theorem, one is based on the K-means score. We analyze the sample complexity of these detectors considering the effects of graph filter's properties, random delays, and other parameters. We show novel applications of the blind detector on robustifying graph learning, identifying antagonistic ties in opinion dynamics, and detecting anomalies in power systems. Numerical experiments validate our findings.
\end{abstract}


\section{Introduction}
A growing  trend in  signal processing, machine learning, statistics is to develop tools for modeling and analyzing data defined on nodes of a graph, formally known as graph signals. Graphs encode the irregularly structured data and model the interactions between adjacent nodes. Such mathematical structure has found applications in social, financial, and biology networks \cite{newman2018networks}. An important problem is to understand the role of graph and its associated dynamics in data. While statistical methods such as graphical models \cite{wainwright2008graphical} have been developed, recently popular GSP models have offered a promising approach that appeals to a wide range of data \cite{intro_GSP1,sandryhaila2013discrete, isufi2022graph}. 

A fundamental building block of GSP is the linear graph filter which extends the classical linear time-invariant (LTI) filter from time domain to the node domain. It enables one to model graph signal as the output of a graph filter subject to excitation input. The graph filter acts as a black box that abstracts the complex network dynamics leading to the observations.
Using this modeling philosophy, GSP algorithms have been developed for signal sampling, interpolation, graph topology learning, etc.~\cite{intro_GSP1}. They have been successfully applied to data from social networks, financial networks \cite{ramakrishna2020user}, brain activity dynamics \cite{brainSingHuang}, and physical networks \cite{intro_learngraphsm2}.

Similar to LTI filters, linear graph filters may be classified as low pass, band pass, and high pass, in accord with its frequency response function \cite{sandryhaila2013discrete}. {In the GSP literature, low pass graph filters are usually assumed which lead to smooth graph signals that have similar signal values over adjacent nodes on the graph. Such \emph{low pass or smooth} signal assumption has been a critical condition used} for graph topology learning \cite{intro_learngraphsm, kalofolias2016learn, dong2019learning, mateos2019connecting}, blind community detection \cite{2018ToBlind, wai2022community, sbmSchaub_2020, roddenberry2020exact}, centrality estimation \cite{intro_blindcentralT, he2022detecting, he2023online}, denoising \cite{denoising}, sampling \cite{sampling}, graph neural networks \cite{intro_NT2019RevisitingGN, wu2019simplifying}.
{On one hand, the low pass assumption can be justified using dynamics models from physics and social science \cite{ramakrishna2020user}. On the other hand, without rigorous validation,
the validity of the low pass assumption can be questionable. Example scenarios include when the system is under attack \cite{sandryhaila2014discrete} or the observed data is corrupted \cite{ferrer2022volterra}.}

{To illustrate the risks in applying the above GSP works on non low pass data, e.g., corrupted data, we present a case study on graph topology learning from meteorology data using GL-SigRep \cite{intro_learngraphsm}. The `clean' dataset consists of the daily mean temperature in Netherlands from June 2020 to Feb 2023 on $N=35$ stations and $M=1000$ days of samples [available: \url{https://www.ecad.eu/}]. A corrupted dataset is formed by inserting noise into the clean dataset at random instances (see the caption of Fig.~\ref{fig:temptuedata}).
The GL-SigRep algorithm \cite{intro_learngraphsm} is applied to learn the weather station graph from the datasets. As the ground truth topology is unknown, we consider two evaluation metrics inspired by \cite{intro_learngraphsm}. We first compare the similarity between the learnt graphs and a proximity graph constructed by assigning an edge between two stations if their distance is less than 125km. The area under ROC (AUROC) scores are $0.8258$ (clean dataset), $0.4379$ (corrupted dataset).
The result for the clean dataset is consistent with the proximity graph, yet for the corrupted dataset, GL-SigRep may have learnt an erroneous graph.
Next, we apply spectral clustering on the learnt graphs where the number of clusters $K$ is determined by {minimum descriptive length (MDL) \cite{wax1985detection}}. As seen in Fig.~\ref{fig:temptuedata}, the clustering results from clean dataset are consistent with geography of the Netherlands which has low, flat lands in the west/north, and higher lands in the east/south; yet for corrupted dataset, the clustering result is inconclusive.}

{The above case study shows that applying GSP works on non low pass data without verifying the assumptions is risky as it may return inconclusive results that can harm downstream applications.     
As noted by \cite{dong2019learning,mateos2019connecting}, a possible solution is to utilize alternative signal models. However this may incur other issues. For example, although the spectral template method \cite{segarra2017network} does not require smooth graph signals for topology learning, it has strict requirements on the excitation model and involves complex optimization criteria.}

{Instead of inventing new signal models, this work aims to \emph{complement existing GSP works} by studying a {blind detection} problem which distinguishes if the given graph signals are low pass or not. Our work is motivated by two important aspects that have not been studied in the GSP community. {\sf(i)} We devise a procedure to validate the low pass property in graph data required by GSP tools. {\sf(ii)} We address a blind system identification problem by determining the type of graph filter followed by the graph data, and thus categorize the dynamics involved. As a preview  {application of our method}, we apply the said procedure (see Sec.~\ref{sec:ex}) to prescreen the corrupted dataset in Fig.~\ref{fig:temptuedata} and then run the GL-SigRep algorithm. The learnt graph achieves an AUROC score of $0.7597$ compared to the proximity graph, and produces a clustering result that is consistent with the geography of the Netherlands.}

\begin{figure*}[t]
\centering
\resizebox{0.32\linewidth}{!}{\sf \includegraphics{fig/fig1-1.png}}\hspace{-.2cm}
\resizebox{0.32\linewidth}{!}{\sf \includegraphics{fig/fig1-2.png}}\hspace{-.2cm}
\resizebox{0.32\linewidth}{!}{\sf \includegraphics{fig/fig1-3.png}}\vspace{-.2cm}
\caption{{\textbf{Graph topology learnt by GL-SigRep \cite{intro_learngraphsm} and clustering results by spectral clustering} on the (Left) clean dataset, (Middle) corrupted dataset, (Right) pre-screened dataset. The corrupted observations are generated as $10$ batches of ideal order-5 high pass signals, each with a duration $M_{\sf ba} = 50$, and they are injected into the dataset at random positions; see \eqref{eq:y-pollute}. `\texttt{X}' denotes the isolated nodes in the learnt graph. We set $m_{\sf batch} = 50$ and $\delta=0.65$ for the pre-screening procedure. The AUROCs compared to the proximity graph are $0.8258$ (Clean) / $0.4379$ (Corrupted) / $0.7597$ (Pre-screened). With the MDL criterion, we found $K=3$ clusters in the graphs learnt from clean and prescreened data.  
}}\vspace{-.3cm}
\label{fig:temptuedata}
\end{figure*}

Note that our blind detection problem is in general ill-posed due to the lack of exact graph topology knowledge. As a remedy, we concentrate on \emph{modular} graphs consisting of a few densely connected components, i.e., a common case for social, biology, finance, physical networks \cite{girvan2002community}. 
Our idea lies in observing a unique \emph{clustered spectral pattern} in the principal subspace spanned by low pass graph signals from modular graphs. Our contributions are:\vspace{-.1cm}
\begin{itemize}[leftmargin=*]
\item We develop blind detection methods for low pass graph signals through identifying the (dis)similarities between the low graph frequency subspace of the graph shift operators (GSO) and the principal subspace spanned by observed graph signals. 
We derive properties of the principal subspace spanned by low pass graph signals.
\item We derive finite-sample performance bounds for the blind detectors when the graph {\sf (i)} contains only one connected component with similar connection probability, {\sf (ii)} is modular and contains $K$ clusters following a stochastic block model (SBM) \cite{2011RoheSBM}. 
{Our result shows that the detection performance improves with more nodes and/or the graph filters have sharper cutoffs.}
We also give insights for the detection performance applied to non-SBM graphs.
\item We discuss applications utilizing the proposed detectors. First, we design a prescreening scheme to robustify graph learning through removing non low pass graph signals from the dataset prior to applying graph learning. Second, we derive models for opinion dynamics with antagonistic relationships \cite{dittrich2020signal, altafini2012consensus} and power systems under data injection attack \cite{Ramakrishnagrid, DrayerAC}.
Our detection method provides a data-driven evidence to discover such phenomena.
\end{itemize}
The proposed detection method takes its inspiration from recent works on blind graph feature learning. Examples are community detection \cite{2018ToBlind, wai2022community, sbmSchaub_2020, roddenberry2020exact}, centrality estimation \cite{intro_blindcentralT, he2022detecting, he2023online}, equitable partitions \cite{scholkemper2022blind}, etc. These works derive knowledge from the principal signal subspace under the premise that the observations are low pass.
In comparison, we take an \emph{inverse problem} perspective by inquiring if a given set of graph signals are low pass. 
Notably, our development involves showing a \emph{converse result} that the said spectral pattern cannot be found in non low pass signals.


Furthermore,
this paper makes a first step towards \emph{blind (topology-free) identification of unknown systems} \cite{tong1994blind} to \emph{determine the type of dynamics} on a network. For example, the presence of antagonistic relationship in social networks result in an opinion formation process that can be described as non low pass graph filter; in power systems, data injection attacks result in graph signals of nodal voltages that are not low pass. 
In comparison, prior works on graph filter identification either utilize controlled perturbation experiments \cite{intro_dynaUniversalitynet, wai2019joint}, or require full knowledge of the graph topology \cite{zhu2020estimating,segarra2016blind}. 


The rest of this paper is organized as follows.
In Sec.~\ref{sec:model}, we describe the observation model and the problem of detecting low pass graph filters. In Sec.~\ref{sec:det}, we propose blind detection algorithm for the low pass graph signals. 
In Sec.~\ref{sec:ana}, we analyze the sampling complexity of the proposed algorithms. Sec.~\ref{sec:ex} describes three new application examples for our algorithm. Sec.~\ref{sec:num} concludes the paper with numerical examples. 
Compared to the conference version \cite{he2021identifying}, we consider low pass graph filters with high cutoff frequencies as well as providing a complete performance analysis. We also studied applications on detecting the  network dynamics types.


\section{Problem Statement} \label{sec:model}

Consider an undirected, connected graph with $N$ nodes given by $G = (V,E)$, where $V = [N] = \{1,...,N\}$ is the node set and $E \subseteq V \times V$ is the set of edges. The (weighted) adjacency matrix of $G$ is a symmetric matrix defined by ${\bm A} = [ A_{ij} ]$ such that $A_{ij} > 0$ if $(i,j) \in E$, otherwise $A_{ij} = 0$. The normalized adjacency matrix is given by $\nA = {\bm D}^{-1/2} {\bm A} {\bm D}^{-1/2}$, where ${\bm D}$ is the diagonal matrix with the $i$th diagonal element $d_i = \sum_{j=1}^N A_{ij}$, and the Laplacian matrix and normalized Laplacian matrix is defined by $\Lap = {\bm D}-{\bm A}$ and ${\bm L}_{\sf norm} = {\bm I} - \nA$ respectively. 
A graph shift operator (GSO), denoted by $\GSO$, is a symmetric matrix such that $S_{ij} \neq 0$ only if $(i,j) \in E$ or $i=j$. For example, the (normalized) adjacency and Laplacian matrices are admissible GSOs. The GSO admits the eigenvalue decomposition $\GSO = \V \bm{\Lambda} \V^\top$, where the columns of $\V = ( \sv_1, \ldots, \sv_N )$ are the eigenvectors and $\bm{\Lambda} = {\rm Diag}( \lambda_1, ... ,\lambda_N )$ is a diagonal matrix of the eigenvalues known as the \emph{graph frequencies}. 
They are assumed to be distinct for simplicity, and 
are in ascending order as
$\lambda_1 < \lambda_2 < \cdots < \lambda_N$ for $\GSO \in \{ \Lap, \nL \}$; or in descending order as $\lambda_1 > \lambda_2 > \cdots > \lambda_N$ for $\GSO \in \{ {\bm A}, \nA \}$.

 {We concentrate on the case where $G$ is a \emph{modular graph} \cite{girvan2002community, newman2013spectral} with $K \ll N$ densely connected clusters\footnote{The number of clusters depends on the nature of the network, e.g., for social networks with $N=100$ to $1000$ nodes, typically $K \leq 5$ \cite{newman2013spectral}, for biological networks with $N=1000$ to $5000$ nodes, typically $K \leq 50$.}. Mathematically, the modularity of $G$ is characterized by the normalized cut: let ${\cal C}_1, \ldots, {\cal C}_K$ be a partition of $[N]$, 
\beq \label{eq:ncut}
{\sf NCut}( {\cal C}_1, \ldots, {\cal C}_K )= \sum_{k=1}^K \sum_{i \in {\cal C}_k} \sum_{j \notin {\cal C}_k } \frac{A_{ij}}{ \sum_{\ell \in {\cal C}_k} \sum_{m \in [N]} A_{\ell m}}. 
\eeq 
If ${\sf NCut}_K^\star := \min_{ {\cal C}_1, \ldots, {\cal C}_K } {\sf NCut}( {\cal C}_1, \ldots, {\cal C}_K ) \approx 0$, we say that the graph is $K$-modular. Importantly, it is known that when $\GSO = {\bm L}_{\sf norm}$, the condition ${\sf NCut}_K^\star \approx 0$ implies $\lambda_1, \ldots, \lambda_K \approx 0$ \cite{von2007tutorial, wai2022community}.
We notice that modular graphs are common in social, economics, biological networks \cite{girvan2002community}. On the other hand, the popular \emph{SBM} is a generative random graph model that lead to modular graphs with high probability \cite{deng2021strong}.}


We consider linear graph filter as a continuous matrix function. 
Let $T_{\sf ord} \in \NN \cup \{ \infty \}$ be the number of filter taps, a linear graph filter is described by \cite{sandryhaila2013discrete}:
\beq \textstyle \label{eq:gf}
\HS = \sum_{t=0}^{T_{\sf ord}} {h^{(t)}} \GSO^t = \V h( \bm{\Lambda} ) \V^\top = \U \bm{h} \U^\top,
\eeq
where we have defined the frequency response function as $h(\lambda) = \sum_{t=0}^{T_{\sf ord}} h^{(t)} \lambda^t$ and $h( \bm{\Lambda} )$ is a diagonal matrix with the entries $\{ h(\lambda_i) \}_{i=1}^N$. 
The diagonal matrix $\bm{h} := {\rm Diag} ( h_1, \ldots, h_N )$ is composed of the eigenvalues for $\HS$ sorted as $|h_1| \geq |h_2| \geq \cdots | h_N|$, such that  $\U$ is the corresponding column re-ordered version of $\V$.

Note that $|h(\lambda)|$ measures amplification or attenuation of the energy of a graph signal at frequency $\lambda$. We observe the following definition for low pass graph filters \cite{ramakrishna2020user}:
\begin{Def} \label{def:lpf}
A graph filter $\HS$ is said to be {$K$ low pass} if 
\beq 
\eta_K := \frac{ \max_{ i=K+1, \ldots, N } |h(\lambda_i)| }{ \min_{ i=1,\ldots,K }  | h(\lambda_i) | } < 1.
\eeq
\end{Def}
\noindent We refer our readers to the illustration in Fig.~\ref{fig:freq}. The smaller $\eta_K$ is, the `sharper' the low pass filter is. The integer $K \in \{1,\ldots, N-1\}$ represents the \emph{cut off frequency}. Notice that the graph frequencies can be unevenly distributed. 
Take $\GSO = \nL$ as an example, a low pass frequency response function $h(\lambda)$ (e.g., a decreasing function in $\lambda$) is usually $K$-low pass if the underlying graph is \emph{modular} with $K$ densely connected components. {Intuitively,} this is because one has $\lambda_1, ..., \lambda_K  {=} 0$ {in the ideal case} while $\lambda_{K+1}, ..., \lambda_N$ are bounded away from $0$ \cite{deng2021strong} and $h(\lambda)$ is continuous in $\lambda$.


\begin{figure}[t]
\centering
{\sf 
\resizebox{!}{.5\linewidth}{\input{fig/fig2}}
}\vspace{-0.1cm}
\caption{Frequency responses $h(\lambda_i)$ against eigenvalue $\lambda_i$. ({\color{green!50.0!black}Green}) ${\cal H}_1({\bm L})= ({\bm I}-0.1{\bm L})^{-1}$, ({\color{blue}Blue}) ${\cal H}_2({\bm L})= ({\bm I}+0.1{\bm L})^{-1}$ and ({\color{red}Red}) ${\cal H}_3({\bm L})= \exp(-0.5{\bm L})$. The red shadow region indicates $\mathbb{L}_{\sf Low}$ region and the black shadow region indicates $\mathbb{L}_{\sf High}$ region. Note that ${\cal H}_2({\bm L})$ and ${\cal H}_3({\bm L})$ are $K=3$-low pass filters while ${\cal H}_1({\bm L})$ is not [cf.~Definition~\ref{def:lpf}].}\vspace{-.4cm} \label{fig:freq}
\end{figure}

We consider the observation model where $M$ filtered graph signals are obtained according to:
\beq \label{eq:y_sec}
{\bm y}_m = \underbrace{\HS \cdots \HS}_{J_m~\text{times}} {\bm x}_m + {\bm w}_m,~m=1,...,M,
\eeq
such that $ {\bm x}_m \in \RR^N$ is the excitation (a.k.a.~input) signal and ${\bm w}_m \in \RR^N$ is the observation or modeling noise. 
{Moreover, the signal component of ${\bm y}_m$ is the result of cascading $J_m$ copies of $\HS$, which models the random delays in observations. This is motivated by the practical consideration where the delays may vary from sample to sample when acquiring the graph signals such as in rumor spreading; see \cite[Sec.~II-C]{zhu2020estimating}.} 
Eq.~\eqref{eq:y_sec} assumes a scenario that the observations/states on nodes following a dynamic process governed by the graph filter $\HS$.
As a convention, we say that ${\bm y}_m$ generated from \eqref{eq:y_sec} is a ($K$-)low pass graph signal with respect to (w.r.t.) $\GSO$ when $\HS$ is ($K$-)low pass. 

For simplified theoretical analysis, we assume that the observed signals are \emph{stationary} \cite{perraudin2017stationary}, i.e., the excitation ${\bm x}_m$ is zero-mean with covariance ${\bm C}_x = \EE[ {\bm x}_m {\bm x}_m^\top ] = \I$. The observation noise is zero-mean with $\EE[ {\bm w}_m {\bm w}_m^\top ] = \sigma^2 {\bm I}$. The delay $J_m$ is uniformly distributed over $\{1,...,J\}$. 
Our detection methods can be readily applied even when ${\bm C}_x \neq {\bm I}$ and is possibly low rank; see the supplementary material.

\subsection{Low Pass Graph Filter Detection Problem}
Given a set of graph signals $\{ {\bm y}_m \}_{m=1}^M$, we wish to determine if the signals are low pass or not w.r.t.~$\GSO$.  
As mentioned, the problem can be easily solved if the GSO $\GSO$ is known. For example, with noiseless observation, $J=1$, we can estimate $|h( \lambda_i )|^2$ for each $i=1,\ldots,N$ through the periodogram ${\bm v}_i^\top \EE[ {\bm y}_m {\bm y}_m^\top ] {\bm v}_i$ \cite{marques2017stationary} and verifying Definition~\ref{def:lpf} directly.

Our aim is to tackle the detection problem in a \emph{blind} setting where the GSO is unknown. This pertains to applications such as graph learning where the graph topology itself is the object to be estimated. However, the blind detection problem is \emph{ill-posed} in general: a set of graph signals can be simultaneously \emph{low pass} w.r.t.~a graph while \emph{not low pass} w.r.t.~another graph. 

\begin{figure}[t]
    \centering
    {\sf 
    \resizebox{.8\linewidth}{!}{\input{fig/fig3}}
    }\vspace{-0.2cm}
    \caption{Low pass (in {\color{green!50!black}green}), non low pass (in black) graph signals generated from $\HS$ on an SBM graph with $K=2$ clusters and $N=50$ nodes. For each graph signal, we plot its elements sorted in  ascending order. The \emph{low pass signals} exhibit a \emph{piecewise-constant} feature.}
    \vspace{-.4cm} \label{fig:y_m}
\end{figure}

Motivated by that many real networks graphs tend to be \emph{modular} \cite{girvan2002community} [cf.~\eqref{eq:ncut}], we inquire if such graphs can lead to any discernible pattern in the low pass and non low pass graph signals generated. Fig.~\ref{fig:y_m} shows the ensemble of low pass, non low pass signals generated on a modular (SBM) graph with $N=50$ nodes. Here, the low pass graph signals generated show a unique pattern \emph{not found} in the non low pass signals: upon a simple sorting, the graph signal is almost \emph{piecewise constant}. 
The demonstrated pattern suggests that \emph{under the premises of modular graph with $K \geq 1$ clusters, where $K$ is known, the blind detection problem for $K$ low pass graph signals can be solvable}. Formally, our main problem is:
\begin{Prob} \label{prob:det}
Given the number of clusters $K$, and a set of graph signals $\{ {\bm y}_m \}_{m=1}^M$ generated from \eqref{eq:y_sec}, we determine if the underlying graph filter $\HS$ is $K$-low pass. 
\end{Prob}
\noindent Denote $\Tgnd \in \{ {\cal T}_0, {\cal T}_1 \}$ as the ground truth hypothesis, where
\begin{itemize}[leftmargin=*, nosep]
    \item ${\cal T}_0$: the null hypothesis refers to when $\HS$ is $K$ low pass;
    \item ${\cal T}_1$: the alternate hypothesis refers to when $\HS$ is \emph{not} $K$ low pass (may be bandpass, high pass, etc.).
\end{itemize}
As an extension, one may wish to consider the alternate hypothesis ${\cal T}_1$ which includes \emph{any graph signals that are not generated by \eqref{eq:y_sec} with $K$-low pass filter} and null hypothesis ${\cal T}_0$ with $K$-low pass signals on non-modular graphs. 
Nevertheless, we found empirically that our detector(s) is effective in detecting signals that are not generated by \eqref{eq:y_sec}.



\section{Detecting Low Pass Graph Signals}\label{sec:det}
We begin our endeavor by analyzing the covariance matrix of \eqref{eq:y_sec}. Observe that
\beq \label{eq:cov_y}
\begin{split}
\CovYn 
& \textstyle = \EE[ {\bm y}_m {\bm y}_m^\top ] = \frac{1}{J} \sum_{ \tau =1}^J [\HS]^{2 \tau } 
+ \sigma^2 {\bm I},
\end{split}
\eeq
where the noiseless covariance is 
\beq \label{eq:cov_y_nl} \textstyle 
\CovY = \frac{1}{J} \sum_{\tau =1}^J [\HS]^{2\tau} = \V \left( \frac{1}{J} \sum_{\tau=1}^J h( \bm{\Lambda} )^{2 \tau} 
\right) \V^\top.
\eeq
Denote the \emph{top-$K$ eigenvectors} of $\CovY$ as $\barU_K$.
If $\HS$ is $K$ low pass, it is clear that $\barU_K$ spans the same subspace as ${\rm span}\{ \sv_1 , \ldots, \sv_K \}$, i.e., the eigenvectors of $\GSO$ corresponding to the $K$ lowest graph frequencies; otherwise if $\HS$ is \emph{not} $K$-low pass, the subspace spanned by $\barU_K$ contains one or more eigenvectors from $\{ \sv_{K+1} , \ldots, \sv_N \}$ of $\GSO$. 

Our idea to handle \Cref{prob:det} is to verify \emph{whether the $K$-dimensional principal subspace of ${\bm C}_y$ coincides with a particular structure in ${\rm span} \{ \sv_1 , \ldots, \sv_K \}$} which we refer to as the \emph{spectral pattern} of low pass signals. 
In the sequel, we show that these structures can be exposed upon careful observations. 
The study is divided into two cases depending on the number of clusters in the graph.
\vspace{.1cm}

\noindent \underline{\textbf{Case of $K=1$:}}
We first focus on the case when it is known that the graph $G$ is \emph{non-modular} (with $K=1$). \Cref{prob:det} in this case aims at detecting graph signals that are $1$-low pass. The development hinges on a well-known property for $\sv_1$:
\begin{Prop}
\label{Prop:lap}
Let $G$ be connected.
    The eigenvector corresponding to the lowest graph frequency, $\sv_1$, is the only positive eigenvector\footnote{Note that both $\sv_1, -\sv_1$ are eigenvectors with the eigenvalue $\lambda_1$. We assume $\sv_1 > {\bm 0}$ to avoid such ambiguity.} of $\nL$, ${\bm L}$, $\nA$, ${\bm A}$. For $i \geq 2$, the eigenvector $\sv_i$ must have  at least one positive and negative element. 
\end{Prop} 
\noindent This proposition is a direct consequence of the Perron-Frobenius theorem; see Appendix~\ref{app:PF}.

Despite its simplicity, Proposition~\ref{Prop:lap} illustrates a \emph{sufficient and necessary} condition for detecting $1$-low pass signals. When $\HS$ is $1$ low pass, the top eigenvector $\su_1$ of the noiseless covariance $\CovY$ is the \emph{only} positive eigenvector; yet when $\HS$ is \emph{not} $1$ low pass, $\su_1$ has at least one positive and negative element. 
Consider the function:
\beq \label{eq:pos}
\Pos( \su ) \eqdef \min\{ \| (\su)_+ - \su \|_1, \| (-\su)_+ + \su \|_1   \},
\eeq
where $( \su )_+ := \max\{ {\bm 0}, \su \}$, and the function outputs zero if and only if $\su$ is a positive (or negative) vector.  
Subsequently, the proposed detector in \Cref{algo:nec} is a direct application of the above principle which checks if $\widehat{\su}_1$ from the sampled covariance $\CovYhat^M$ is the most positive eigenvector.\vspace{.1cm}


\begin{algorithm}[t]
\caption{Tackling \Cref{prob:det}}\label{algo:nec}
\begin{algorithmic}[1]
\STATE \textbf{INPUT}: Observed graph signals $\{ {\bm y}_m \}_{m=1}^M$, no.~of clusters in the graph $K$, threshold parameter $\delta$.
\STATE Evaluate $\CovYhat^M = (1/M)\sum_{m=1}^M {\bm y}_m {\bm y}_m^\top$.
\STATE Compute the eigenvalue decomposition (EVD) of $\CovYhat^M$ as $\hatU \widehat{\bm{\Lambda}} \hatU^\top$, where the eigenvectors $\hatsu_1, ..., \hatsu_N$ are sorted in descending order with the eigenvalues. 
\STATE {\bf If $K=1$,} perform the detection as
\begin{equation} \label{eq:k1_algo}
\hatTnec = \begin{cases}
{\cal T}_0, & \text{if}~\Pos ( \hatsu_1 ) < \min_{j=2,...,N} \Pos ( \hatsu_j ),\\
{\cal T}_1, & \text{otherwise}. 
\end{cases} 
\end{equation} 
\STATE {\bf If $K \geq 2$,} perform the detection as
\beq \label{eq:det_suff}
\hatTsuff = \begin{cases}
{\cal T}_0, & \text{if}~
\Kmeans( [\hatsu_1, \ldots, \hatsu_K] )< \delta,
\\
{\cal T}_1, & \text{otherwise}, 
\end{cases} 
\eeq
where $\Kmeans( \cdot )$ is defined in \eqref{eq:kmeansUK}.
\STATE \textbf{OUTPUT}: the estimated hypothesis $\hatTnec$. 
\end{algorithmic}
\end{algorithm}

\noindent \underline{\textbf{Case of $K \geq 2$:}}
We consider the case when the graph is modular with $K$ densely connected clusters. \Cref{prob:det} in this case refers to detecting $K$-low pass graph signals.
{For simplicity, we fix $\GSO = \nL$ as the GSO in the following discussions. Notice that both our detection method and analysis can be applied to other choice of GSOs}.

We study  $\{ {\bm v}_1 , \cdots , {\bm v}_K \}$ \emph{as well as} the bulk eigenvectors $\{ \sv_{K+1}, \ldots, \sv_N \}$. Unlike the previous case, the spectral pattern of $\{ {\bm v}_2 , \cdots , {\bm v}_K \}$ may not be obvious at the first glance, where the eigenvectors' elements fluctuate between positive and negative values. 
Fig.~\ref{fig:VN-k} (left) shows the \emph{sorted} and \emph{unsorted} eigenvectors of $\GSO = \nL$ for a modular (SBM) graph with $K=3$ clusters. \emph{Upon sorting}, the principal eigenvectors $\sv_1, \sv_2, \sv_3$ exhibit a piece-wise constant behavior; while it is not the case for bulk eigenvectors $\sv_{4}, \sv_5, ...$. We should mention that the above phenomena has been observed \cite{von2007tutorial} for modular graphs where for $\ell = 1,\ldots,K$, the elements of ${\bm v}_\ell$ will be `flat' over the cluster, i.e., $[ {\bm v}_\ell ]_{i,:} \approx [ {\bm v}_\ell ]_{j,:}$ for all nodes $i,j$ belonging to the same cluster; and it has been shown analytically for stochastic block models \cite{2011RoheSBM,yun2014accurate, lei2015consistency,deng2021strong}. 

\begin{figure}[t]
\centering
\subfigure{\resizebox{0.455\linewidth}{!}{\sf \includegraphics{fig/fig4-1.pdf}}
}%
\subfigure{\resizebox{0.44\linewidth}{!}{\sf \includegraphics{fig/fig4-2.pdf}}}\vspace{-.2cm}

\caption{(Left) Eigenvectors of $\nL$ where the underlying graph contains $K=3$ clusters. Paled colors denote unsorted eigenvectors. (Right) K-means score $\Kmeans (\V_{K})$ and that of last $N-K$ eigenvectors $(N-K)^{-1} \sum^{N}_{i=K+1}\Kmeans(\sv_{i})$ against $N/K$. 
}\vspace{-.4cm}
\label{fig:VN-k}
\end{figure} 


Together with \eqref{eq:cov_y_nl}, our observations suggest that low pass signals lead to covariance matrix with \emph{clusterable principal eigenvectors}. To detect such pattern, a natural choice is to apply the $K$-means score to measure the degree of (non) low pass-ness.
For any matrix $\N \in \RR^{N \times K}$, define {the score}:\vspace{-.1cm}
\beq \label{eq:kmeansUK}
\Kmeans( \N ) :=\underset{ \begin{subarray}{c} S_i \cap S_j = \emptyset, i \neq j \\ S_1 \cup \cdots \cup S_K = V \end{subarray} }{ \min } \sum_{i=1}^K \sum_{\ell \in S_i}\left\| {\bm n}_{\ell}- \frac{1}{\left|S_i\right|} \sum_{j \in {\cal S}_i} {\bm n}_{j} \right\|^2, \vspace{-.1cm}
\eeq
where ${\bm n}_{\ell} \in \RR^K$ is the $\ell$th row of $\N$. {The minimization can be effectively handled by the standard $K$-means procedure \cite{kumar2004simple}}. 

We observe that when $\HS$ is $K$-low pass, one has ${\rm span} \{ \hatU_K \} \approx {\rm span} \{ {\bm v}_1 , \ldots, {\bm v}_K \}$ and thus $\Kmeans( \hatU_K )$ is small; vice versa, when $\HS$ is not $K$-low pass, $\hatU_K$ contains at least one eigenvector from the bulk $\{ {\bm v}_{K+1} , \ldots, {\bm v}_N \}$ and the score $\Kmeans( \hatU_K )$ is large as at least one of the eigenvectors is not clusterable. Finally, this suggests a threshold detector on $\Kmeans( \hatU_K)$ in \Cref{algo:nec}. Although the detector shares the same ingredient with clustering algorithms such as spectral clustering \cite{von2007tutorial} in applying the $K$-means score metric \eqref{eq:kmeansUK}, our goals are different. We aim to measure how the principal eigenvectors of $\CovYhat^M$ are aligned with ${\rm span} \{ \sv_1, \ldots, \sv_K \}$, via \eqref{eq:kmeansUK} which is akin to performing a clusterability analysis \cite{ben2008measures}. 

Thus far, we have only presented empirical insights to derive the blind detector. 
To justify the correctness of \Cref{algo:nec}, it is necessary to impose further structure on the detection problem.
Below, we analyze an \emph{idealized} modular graph model of $G$ yielded by the planted partition SBM. 
\begin{Def} \cite{2011RoheSBM} Let $p, r > 0$, $p + r \leq 1$. We denote $G \sim {\sf SBM}( K,N, r,p)$ as the random graph with $N$ nodes that are partitioned into $K$ equal sized blocks ${\cal C}_1, \ldots, {\cal C}_K$\footnote{For simplicity, assume that $N$ is divisible by $K$.}. The edges are generated independently and randomly according to:
\beq 
{\rm Pr}[ (i,j) \in E ] = \begin{cases}
p+r &,~\text{if}~i \in {\cal C}_k, j \in {\cal C}_k, \\
r &,~\text{if}~i \in {\cal C}_k, j \in {\cal C}_\ell,~k \neq \ell.
\end{cases}  
\eeq 
\end{Def}
\noindent Take $\GSO = \nL$.
We shall verify two properties: (i) $\Kmeans( \N ) \approx 0$ when $\N$ is a column permuted version of $\V_K$, (ii) $\Kmeans ( \N )$ is bounded away from zero when $\N$ contains at least one eigenvector from $\{ \sv_{K+1}, \ldots, \sv_N \}$. 

The first property is confirmed by the proposition:
\begin{Prop}
 \label{the:kmeans}
Let $ G \sim {\sf SBM}(K, N, r, p)$ with $p \geq r > 0$, $\frac{p}{K}+r \geq \frac{32\log N+1}{N}$, and take $\GSO = \nL$ with unweighted adjacency matrix. Then,
with probability at least $1-2 / N$,\vspace{-.1cm}
\begin{equation} 
\begin{aligned}
\Kmeans(\N ) & 
\leq\frac{35^2 K^3 \log N}{  p( N-K ) } 
\end{aligned}\vspace{-.1cm}
\end{equation}
for any $\N$ given by permuting the columns of $\V_K$.
\end{Prop}
\noindent The proof, which is due to \cite{sbmyu2014impact}, can be found in \Cref{app:vpop}. 

For the second property where the columns of $\N$ contain at least one eigenvector from $\{ \sv_{K+1}, \ldots, \sv_N \}$, we  provide a partial answer motivated by empirical studies. 
To this end, in Fig.~\ref{fig:VN-k} (right)
we simulate the $K$-means scores on the eigenvectors $\Kmeans(\sv_\ell)$, $\ell \in 1, \ldots, N$, of $\nL$ averaged from {$M=500$} realizations of $G \sim {\sf SBM}( K, N, \log(N)/N, 0.1 )$. Observe that $\Kmeans( \V_K )$ decreases as ${\cal O}(1/N)$, while for $\ell=K+1,\ldots,N$, $\Kmeans( \sv_\ell )$ remains bounded from below by a constant. This leads to the following assumption:
\begin{assumption} \label{conj:eig}
For $N \gg K$. Let $G \sim {\sf SBM}(K, N, r, p)$ and denote $\sv_\ell$ as the $\ell$th largest eigenvector of the corresponding $\nL$. There exists a constant ${\rm c}_{\sf SBM}$ independent of $N,p,r$ such that $\Kmeans( \sv_\ell ) \geq {\rm c}_{\sf SBM} > 0$, $\ell= K+1,\ldots,N$.
\end{assumption}
\noindent 
To our best knowledge, proving the above assumption is an open research problem. 
We refer the readers to works on bulk eigenvectors in SBMs for partial results that affirm the conjecture: \cite{kadavankandy2015characterization,bai2012limiting} applied random matrix theory to show that the elements of the {bulk eigenvectors} $\sv_\ell$, $\ell=K+1,\ldots,N$ follow a near-Gaussian distribution that would lead to \Cref{conj:eig}. 

Now, suppose that at least one of the column vectors in $\N$ is a bulk eigenvector, $\sv_\ell$, $\ell \in \{ K+1, \ldots, N \}$, of $\nL$. Under \Cref{conj:eig}, we have
\beq 
\begin{split} 
\Kmeans( \N ) & \geq \Kmeans ( {\bm n}_{1}^{\rm col} ) + \cdots + \Kmeans ( {\bm n}_{K}^{\rm col} ) \\
& \geq \Kmeans ( {\bm n}_i^{\rm col} ) = \Kmeans ( \sv_\ell ) \geq {\rm c}_{\sf SBM} > 0,
\end{split} 
\eeq
where ${\bm n}_i^{\rm col}$ denotes the $i$th column of $\N$. 
Combining with properties (i) and (ii), we conclude that as $\hatU_K$ can be modeled with the above cases for $\N$ when $M \to \infty$ and $\sigma \to 0$, thresholding on $\Kmeans ( \hatU_K )$ detects if $\HS$ is low pass or not. 

\begin{Remark}\label{remark:K}
Solving \Cref{prob:det} requires knowledge on the number of clusters $K$. {For a number of applications, the latter is known a-priori, e.g., the graph of US Senate is believed to have $K=2$ due to the bipartisan nature. Otherwise, we apply a heuristic procedure extended from Algorithm~\ref{algo:nec}: set $\hat{K} = 1$,
\begin{enumerate}
    \item If $\hat{K}^{-1} \Kmeans ( \widehat{\U}_{\hat{K}} ) < \delta_{ \sf Est }$, then declare that the graph filter is $\hat{K}$ low pass; otherwise, go to step 2.
    \item If $\hat{K} > K_0$, declare that the graph filter is not low pass; otherwise, set $\hat{K} = \hat{K} + 1$.
\end{enumerate}
In the above, $\delta_{ \sf Est }$, $K_0$ respectively denote the threshold parameter and estimate on the maximum number of clusters.}
\end{Remark}

\section{Finite-Sample Performance Analysis} \label{sec:ana} \vspace{-.1cm}
This section analyzes the performance of proposed algorithm when the number of samples $M$ is finite. 
In a nutshell, our analysis relies on the fact that the eigenvectors of $\CovY$ are re-ordered version of the GSO's eigenvectors [cf.~\eqref{eq:cov_y_nl}].
It suffices to estimate the number of samples required for the spectral pattern to emerge via the Davis-Kahan theorem. 

To facilitate our analysis, let us consider the assumption:
\begin{assumption} \label{ass:dis}
The magnitudes of frequency response  are distinct at all graph frequencies, i.e., $|h(\lambda_i)| \neq |h(\lambda_j)|$ for all $i \neq j$.
\end{assumption}
\noindent This assumption is easy to satisfy, e.g., it holds when $\lambda_i \neq \lambda_j$ and $h(\cdot)$ is a strict monotonic function. Under ${\cal T}_0$, the matrix $\barU_K$ is a column-permuted version of $\V_K = ( {\bm v}_1 \cdots {\bm v}_K )$. 
Let us also define a few quantities on the eigenvalues of $\CovY$. Set as the $j$th eigengap for $\CovY$
\beq \label{eq:eigBdef}
\eigB_j := \beta_j(\CovY)-\beta_{j+1}(\CovY),~j=1,\ldots,N-1,
\eeq 
and $\eigB_0 = \infty$, where $\beta_j(\CovY)$ denotes the $j$-th largest eigenvalue of $\CovY$. 
To get insights, we recall that $h_j$ is the $j$th largest entry in the set of filter frequency responses $\{ |h(\lambda_1)|, \ldots, |h(\lambda_N)| \}$. Observe that 
\beq \textstyle \label{eq:eigBj_bd}
\eigB_j = \frac{1}{J} \sum_{\tau =1}^J \big( h_j^{2 \tau} - h_{j+1}^{2 \tau} \big),
\eeq 
where the $j$th eigengap depend on the frequency response of the graph filter. For example with $K$ low pass filter defined on $K$ clusters graphs, as the frequency response maybe flat before the cutoff frequency at $\lambda_K$, we expect $\eigB_K$ to be large, and is in fact, proportional to the reciprocal of low pass ratio $1/\eta_K$, while $\{ \eigB_{j} \}_{j=1}^{K-1}$ is small. In other words, these constants directly measure the `sharpness' of the graph filter across different graph frequencies.

\vspace{.1cm}
\noindent \underline{\textbf{Case of $K=1$.}}
In this case, \Cref{algo:nec} relies on the positivity function \eqref{eq:pos} applied to the top eigenvector of $\CovYn$. 
We observe the sampling complexity bound: 
\begin{Theorem}\label{theo:K1T}
Assume that there exists constants ${\rm c}_0, \bareig$ such that the graph filters of interest (under ${\cal T}_0$ or ${\cal T}_1$) satisfy
\beq \label{eq:spec_gap_1}
\begin{split} 
0 < {\rm c}_0 & \leq \min_{j=2,\ldots, N} {\sf Pos} ( {\bm v}_j ),~~0 < \bareig \leq \min_{ j=1,\ldots,N-1 } \eigB_j 
\end{split} 
\eeq 
and the noise variance satisfies $\sigma^2\sqrt{N} < 2^{-2.5} {\rm c}_0 \bareig$.
If the number of samples $M$ satisfies
\beq \label{eq:sample_complexity_1}
\sqrt{\frac{M}{\log{M}}} \geq  
\frac{2 {\rm c}_{1} \operatorname{tr}(\CovYn)}{2^{-2.5} { {\rm c}_0 \bareig } / {\sqrt{N}}-\sigma^{2} },
\eeq
where ${\rm c}_1$ is a constant independent of $N,M$, then  it holds
${\rm Pr} ( \hatTnec = \Tgnd ) \geq 1 - 10/M$,
where the randomness is due to sampling from \eqref{eq:y_sec}.
\end{Theorem}
\noindent The proof is relegated to \Cref{app:pftheoK1T0}.

The constants ${\rm c}_0, \eigap_{\sf min}$ inform the class of admissible (low pass or non-low pass) graph filters and its underlying GSO that are detectable by \Cref{algo:nec} with enough samples.
In particular, the constant ${\rm c}_0$ denotes the spread of the values in non-principal eigenvector $\sv_j$ of admissible GSOs. There are normally $\Theta(N)$ negative elements in $\sv_j$ as deduced in \cite{Fiedler1975}, implying ${\rm c}_0=\Theta(\sqrt{N})$. Thus, the denominator is positive with $\sigma^2 = {\cal O}( \bareig )$. 
On the other hand, the constant $\bareig$ depends on the class of admissible graphs and the `sharpness' of the graph filters defined on them as discussed in \eqref{eq:eigBj_bd}. It is bounded away from zero since the graph  only has $K=1$ cluster. 
We deduce that the number of samples $M$ required for correct detection decreases with small $\sigma$, ${\rm tr}( \CovYn )$.

\vspace{.1cm}
\noindent \underline{\textbf{Case of $K\geq 2$.}}
The analysis for detecting low pass graph filters with high cutoff frequency ($K \geq 2$) is more involved since the graph filters may have frequency response that are non-monotone in $\lambda$. 


We obtain the sampling complexity bound for \Cref{algo:nec}:
\begin{Theorem}\label{theo:Km_sample}
Let $G \sim {\sf SBM}(K, N, r, p)$ with $p \geq r > 0$, $\frac{p}{K}+r \geq \frac{32\log N+1}{N}$ and take $\GSO = \nL$ with unweighted adjacency matrix. Accordingly, we denote the classes of $K$-low pass and non $K$-low pass graph filters as ${\cal H}_{\GSO}^{\sf low}, {\cal H}_{\GSO}^{\sf hi}$, respectively. Assume that 
\beq \label{eq:spec_gap_K}
\bareig \textstyle := \inf_{ \eigB_{K} : \HS \in {\cal H}_{\GSO}^{\sf low} \cup {\cal H}_{\GSO}^{\sf hi} }~ \eigB_{K} > 0, 
\eeq 
the following threshold-dependent constant:
\beq \label{eq:tildedelta}
\tilde{\delta}_{\sf min} :=
\min \left\{\sqrt{\frac{\delta}{2}- \frac{1225 K^{3} \log N}{p (N-{K})}},\sqrt{{\rm c}_{\sf sbm}} -\sqrt{\delta} \right\} > 0,
\eeq
and the noise variance satisfies $\sigma^{2}\leq \tilde{\delta}_{\sf min} {\bareig} / {\sqrt{8K}}$. If the number of samples $M$ satisfies
\beq \label{eq:sample_complexity_K}
\sqrt{\frac{M}{\log{M}}} \geq  
\frac{2 {\rm c}_{1} \operatorname{tr}(\CovYn)}{\tilde{\delta}_{\sf min}\bareig / \sqrt{8K}-\sigma^{2} },
\eeq
where ${\rm c}_1$ is a constant independent of $N,M$ to be defined later in \Cref{lem:sample}, then it holds $ {\rm Pr}( \hatTsuff = \Tgnd ) \geq 1 - 10/M - 2/N$,
where the randomness is due to sampling from \eqref{eq:y_sec} and $G \sim {\sf SBM}(K, N, r, p)$.
\end{Theorem}
\noindent The proof is relegated to Appendix~\ref{app:pftheoKm_sam}. {Note that when $K=1$, $\bareig$ defined in \eqref{eq:spec_gap_K} also satisfies the second inequality in \eqref{eq:spec_gap_1}.}

The constants $\bareig, \tilde{\delta}_{\sf min}$ characterize the class of $K$ low pass or non $K$ low pass graph filters detectable by \Cref{algo:nec} upon accruing a sufficient number of samples. In specific, $\tilde{\delta}_{\sf min}$ depends on the user defined threshold value $\delta$ in \Cref{algo:nec}. If $N$ is large, and the clusters are dense (due to $p$) with sparse inter-cluster connections (due to $r+p \leq 1$), we can set $\delta = \Theta(1)$ which leads to $\tilde{\delta}_{\sf min} = \Theta(1)$.
On the other hand, $\bareig$ in \eqref{eq:spec_gap_K} depends on the `sharpness' of the admissible graph filters. 
As the sample complexity \eqref{eq:sample_complexity_K} is inversely proportional to $\bareig$, less samples will be required to correctly detect the graph filter if the admissible graph filters are `sharper' at the cutoff frequencies. With a sufficiently small noise variance and large $N, M$, \Cref{algo:nec} is guaranteed to solve \Cref{prob:det}. 

{Finally, we estimate the sampling complexity of \Cref{algo:nec} from \eqref{eq:sample_complexity_1}, \eqref{eq:sample_complexity_K}. Particularly, with sufficiently small noise $\sigma^2$ and large $N$, our theorems show that the \emph{minimum number of samples required for correct detection with high probability} satisfies:
\beq 
\sqrt{ M / \log M } = \Omega \big( \, \bareig^{-1} \sqrt{K} \, \big) .
\eeq 
As $\bareig$ is large for `sharp' graph filters, i.e., it is proportional to $1/\eta_K$ for small $\eta_K$ with low pass graph filters, we anticipate that the detection performance improves for {\sf(i)} graphs with few number of clusters and {\sf(ii)} the graph filters to be detected have sharp cutoffs. 
}

\begin{Remark}
The above result has concentrated on $G \sim {\sf SBM}(K,N,r,p)$ and $\GSO = {\bm L}_{\sf norm}$. As mentioned, the choice of SBM model illustrates the performance of \Cref{algo:nec} in an ideal model. In practice, it gives insights that the detection performance improves when the underlying graph has dense clusters and sparse inter-cluster connections. 
\end{Remark}


\section{Application Examples}\vspace{-.1cm} \label{sec:ex}
Tackling \Cref{prob:det} is a critical step for downstream GSP applications pertaining to graph data. This section describes examples to illustrate the applications of proposed algorithms in {\sf (i)} robustifying graph learning, {\sf (ii)} detecting the sign-ness of opinion dynamics in social networks, and {\sf (iii)} detecting anomalies for power networks.

\subsection{Robustifying Graph Learning}
Graph learning is a longstanding problem in GSP, whose aim is to infer the graph topology from graph signal observations \cite{mateos2019connecting, dong2019learning}. Among others, a popular setup is to model observations as \emph{smooth graph signals}, i.e., low pass graph signals \cite{kalofolias2016learn, dong2016learning}.
In reality, the observations can occasionally be corrupted by \emph{outliers}. 
In light of this, robust graph learning algorithms have been recently proposed, e.g., \cite{berger2020efficient} for an inpainting approach, \cite{wang2022distributionally} applied distributional robust optimization, and \cite{araghi2023outlier} for an outlier-resilient algorithm. 

Our idea is to model the graph signals as being occasionally corrupted with a non low pass signal ${\bm p}_m$. Let $\HS$ be a low pass graph filter, for $m=1,\ldots,M$,
\beq \label{eq:y-pollute}
{\bm y}_m = \begin{cases}
{\cal H}( \GSO ) {\bm x}_m + {\bm w}_m &,~m \in {\cal M}_{\sf clean}, \\
 { {\bm p}_m }
+ {\bm w}_m &,~m \in {\cal M}_{\sf pollut},
\end{cases}
\eeq 
where ${\cal M}_{\sf clean}, {\cal M}_{\sf pollut}$ partitions the set $[M] = \{1, \ldots, M \}$, and ${\bm w}_m$ is the modeling/observation noise. Notice that \eqref{eq:y-pollute} is inspired by \cite{berger2020efficient} when the sampling set is small and selected randomly. Since the outlier signals can be due to anomalous sensor measurements or missing data \cite{ferrer2022volterra}, \eqref{eq:y-pollute} models the practical scenario when these events persist for a short period of time and recover afterwards. 
It is suggested to model the outlier signals as non low pass (i.e., high pass) signals on graph filters w.r.t.~possibly time varying GSO \cite{ramakrishna2020user}. 
 {For example, the latter can be modeled as ${\bm p}_m = {\cal H}_{\sf HP}( \GSO ) {\bm x}_m$, where ${\cal H}_{\sf HP}(\cdot)$ corresponds to a non low pass filter violating Definition~\ref{def:lpf}; or it can be modeled as a contaminated signal ${\bm p}_m = {\cal H}( \GSO ) {\bm x}_m + \Delta {\bm p}_m$ with $\Delta {\bm p}_m$ being a sparse vector.}

From Fig.~\ref{fig:temptuedata} (middle), we recall that learning the graph topology directly from \eqref{eq:y-pollute} via methods like GL-SigRep \cite{dong2016learning} can produce inconclusive result. 
To this end, we propose a simple solution that pre-screens the dataset by removing corrupted signals using \Cref{algo:nec}. 
Let $m_{\sf batch}$ be the batch size such that there are $M_{\sf batch} = M / m_{\sf batch}$ batches. Assuming that $K$ is known, we apply
\begin{enumerate}
    \item For $b = 1, \ldots, M_{\sf batch}$, apply Algorithm~\ref{algo:nec} on ${\bm y}_m$, $m = (b-1) m_{\sf batch} + 1, \ldots, b m_{\sf batch}$. Remove the batch if the involved signals are not $K$ low pass.
    \item Apply graph learning method such as GL-SigRep \cite{dong2016learning} or SpecTemp \cite{segarra2017network} on the remaining graph signals.
\end{enumerate}
Fig.~\ref{fig:temptuedata} (right) demonstrated the efficacy of the pre-screened graph learning procedure (via GL-SigRep) above. 
 {We remark that the parameter $m_{\sf batch}$ trades off between the accuracy of pre-screening and performance of graph learning. The success of the procedure hinges on whether ${\cal M}_{\sf pollut}$ is a rare event or not; see Sec.~\ref{sec:num} for a detailed study.}

\subsection{Detecting Signed Opinion Dynamics}
Signed graph is a common model for social networks with cooperative (trust/friendly) and antagonistic (distrust/hostile) relationships \cite{dittrich2020signal}. We set $G_s = \left( {  V}, {E}^{+}, {E}^{-} \right)$, where $V = \{1,\ldots, N\}$ denotes the set of agents, $E^{+}, E^- \subseteq V \times V$ denote the positive, negative edge sets such that $E^+ \cap E^- = \emptyset$.
Accordingly, the graph is endowed with two adjacency matrices ${\bm A}_{E^{+}}$ and ${\bm A}_{E^{-}}$, such that $[{\bm A}_{E^{+}}]_{i,j} > 0$ iff $(i,j) \in E^{+}$ and $[{\bm A}_{E^{-}}]_{i,j} < 0$ iff $(i,j) \in E^{-}$, otherwise $[{\bm A}_{E^{+}}]_{i,j} = [{\bm A}_{E^{-}}]_{i,j} = 0$. 
{We notice that signed graph learning can be achieved using classical algorithms such as \cite{friedman2008sparse} through accounting anti-correlations in data, and recent papers have proposed improved algorithms on the topic \cite{dinesh2022efficient, matz2020learning, karaaslanli2022scsgl}.}

From a system identification point of view, a relevant problem is to inquire \emph{if a social network is dominated with antagonistic relationships through observing opinion data}. We demonstrate that this problem can be approximated as a special case of \Cref{prob:det}, and therefore \Cref{algo:nec} can be applied. 
Our model is based on that of Altafini \cite{altafini2012consensus} for the opinion formation process (see \cite{proskurnikov2018tutorial} for a survey involving signed graphs) together with $R$ stubborn agents whose opinions are unaffected by others \cite{stubborn1,stubborn3}. The latter represent entities who inject opinions into the social network.
Let ${\bm y}_m(\tau)$, ${\bm z}_m(\tau)$ be opinions of regular, stubborn agents at time $\tau$, we have: 
\beq \label{eq:opi_sign}
{\bm y}_m(\tau+1) = \alpha ({\bm A}_{E^+} + {\bm A}_{E^-}){\bm y}_m(\tau) + (1-\alpha) \B{\bm z}_m,
\eeq 
where $\alpha \in (0,1)$ and $\B \in \mathbb{R}^{N\times R}$ represents the mutual trust/distrust between stubborn and regular agents. Assume normalized weighted adjacency matrices such that $|{\bm A}_{E^+}+{\bm A}_{E^-}|{\bm 1} = {\bm 1}$ and $[\alpha({\bm A}_{E^+}+{\bm A}_{E^-}), (1-\alpha)\B]{\bm 1} = {\bm 1}$.
The recursion \eqref{eq:opi_sign} is stable and admits an equilibrium state: 
\beq\label{eq:signedeq}
{\bm y} =(1-\alpha)(\I-\alpha({\bm A}_{E^+} + {\bm A}_{E^-}))^{-1}\B {\bm z}.
\eeq

Notice that our signal model in Sec.~\ref{sec:model} was defined only for \emph{unsigned graphs}. To analyze \eqref{eq:signedeq}, we utilize an \emph{unsigned surrogate adjacency matrix for $G_s$} as the GSO matrix $\GSO = \overline{{\bm A}}={\bm A}_{E^+}-{\bm A}_{E^-}$, representing an unsigned surrogate graph $G = (V, E^+ \cup E^-)$.
We observe three cases:
\begin{enumerate}[leftmargin=*]
\item When all edges are positive, \eqref{eq:signedeq} is the output of graph filter $\HS =(\I-\alpha \GSO )^{-1}$ with the excitation $(1-\alpha)\B{\bm z}$. These are $K$-low pass graph signals \cite{ramakrishna2020user}.
\item When all edges are negative, \eqref{eq:signedeq} is the output of graph filter $\HS =(\I+\alpha \GSO )^{-1}$. In this case, the graph signals are \emph{not $K$-low pass}.
\item When there are both positive and negative edges, it can be shown using a simple Taylor approximation that
\beq
\begin{aligned}
\HS &=(\I-\alpha( \GSO +2|{\bm A}_{E^-}|))^{-1}\\ 
&\approx(\I-\alpha \GSO )^{-1}+\mathcal{O}(\alpha \|{\bm A}_{E^-}\|).
\end{aligned}
\eeq
The number/strength of negative edges influence the low pass property of $\HS$. 
\end{enumerate}
Recall that the metrics ${\sf Pos}(\widehat{\su}_1)$ or $\Kmeans( \widehat{\U}_K )$ in \Cref{algo:nec} estimates the level of \emph{low-pass-ness} for a set of graph signals. Combined with the observations above, they  measure the strength of \emph{antagonistic relationships} in a social network.

 {We remark that our approach differs from recent studies on signed GSP, e.g., \cite{dinesh2022efficient}. For given graph signals, our aim is to detect if the graph is \emph{signed or not} without knowing the graph topology, while signed GSP considers the frequency analysis of the signals defined on the signed graph, where the latter topology and the signs of edges are known a-priori.
}

\subsection{Detecting Anomalies in Power Networks}
An important task in power system is to protect the latter against \emph{false data injection attack} (FDIA) which may obfuscate power system state estimator and lead to unstable behavior. As the power system operations are dependent on the grid of transmission lines and buses, GSP-derived anomaly detectors have been studied in a number of prior works \cite{Ramakrishnagrid, DrayerAC}. Most of these detectors require knowledge of the grid's topology. 
While the latter's knowledge is usually available since the power system is man-made, in some scenarios, its topology may not be precisely estimated when the grid is energized \cite{cavraro2017voltage}, e.g., some power lines may not be operational. 

Similar to the previous application example, our aim is to showcase that FDIA detection can be cast as a special case of \Cref{prob:det}, where we relate the attack-free system states as \emph{low pass graph signals} and FDIA event as \emph{non low pass signals} \cite{Ramakrishnagrid}. Subsequently, \Cref{algo:nec} can be applied regardless of error in topology estimation.
To fix idea, a power network is described by $G=(V,E)$ such that $V$ is the set of buses, and $E$ is the set of transmission lines between buses. Let ${\bm Y} \in \mathbb{C}^{N \times N}$ be the admittance matrix such that $Y_{i,j} = 0$ if $(i,j) \notin E$. 
In quasi-steady state at time $t$, we observe the voltage phasors $\sv_t$ which can be approximated as the output of a low pass graph filter ${\cal H} (\boldsymbol{Y})$ with the excitation $\boldsymbol{x}_t$
\cite[Lemma 1]{Ramakrishnagrid}:
\beq \label{eq:nofdi}
\boldsymbol{v}_{t} \approx {\cal H} (\boldsymbol{Y}) {\bm x}_t +\boldsymbol{w}_{t} = (d \mathbf{I}+\boldsymbol{Y})^{-1}\boldsymbol{x}_{t}
+\boldsymbol{w}_{t},
\eeq 
where $d\in \mathbb{C}$ is a scalar that depends on ${\bm Y}$ and $\boldsymbol{w}_{t} \in \mathbb{C}^N$ captures the slow time-varying nature of the model.
On the other hand, the observed ${\bm v}_t$ under FDIA is:
\beq \label{eq:yesfdi}
\boldsymbol{v}_{t,{\sf FDI}} \approx \mathcal{H}(\boldsymbol{Y})  \boldsymbol{x}_{t} + \bm{\delta}_t
+ \boldsymbol{w}_{t} ,
\eeq 
where $\bm{\delta}_t$ models a possibly sparse attack signal \cite{DrayerAC}, i.e., not low pass in general. 
Lastly, under the assumption that an FDIA event will persist for several samples, \Cref{algo:nec} can be applied to batches of the voltage phasor graph signals to yield a \emph{topology-free} FDIA detector. 
We envision that such a detector may be used in conjunction with existing algorithms such as \cite{Ramakrishnagrid, DrayerAC} for efficient detection of FDIA.


\section{Numerical Experiments}\label{sec:num}

\subsection{Detecting Low Pass Graph Signals}
We consider synthetic graph signals on generated random graphs to verify our analysis result. The graph signals are generated according to \eqref{eq:y_sec} with the excitation given 
by ${\bm x}_m\in \mathbb{R}^N \sim N(\mathbf{0}, \I)$ 
and the observation/modeling noise follows ${\bm w}_m \sim N\left(\mathbf{0}, \sigma^{2} \I\right) $.
For the experiments on non-modular graphs (with $K=1$), we generate $G$ as an Erdos-Renyi graph with connection probability of $p_{\sf er}=2 \log (N) / N$; for experiments with $K \geq 2$ clusters, we generate the graphs according to $G \sim {\sf SBM}(K,N,\log(N)/N,4 \log(N)/N)$. To benchmark the detection performance, 
we consider the classes for low pass or non low pass graph filters $\HS$ with the hypothesis ${\cal T}_0:e^{-\tau \nL}$ and $
{\cal T}_1: e^{\tau \nL}$, where $\tau > 0$ is a  parameter controlling the sharpness of the filters.
We also compare with a two-step detection scheme based on SpecTemp \cite{segarra2017network} (via the fast implementation by \cite{shafipour2019online}) which first learns a GSO from the stationary graph signals, and then verify if the graph signals are low pass using Definition~\ref{def:lpf}. 
All experiments are conducted with $100$ Monte-Carlo trials.

The first experiment concentrates on the case with $K=1$ and evaluates the missed detection (MD) rate, ${\rm Pr} (\widehat{\mathcal{T}}=\mathcal{T}_{1} \mid \mathcal{T}_{0})$, and false alarm (FA) rate, ${\rm Pr}(\widehat{\mathcal{T}}=\mathcal{T}_{0} \mid \mathcal{T}_{1})$ against $M$ and $N$ respectively in Fig.~\ref{fig:algo1} (left) and (right). We fixed $(N,\sigma^{2})=(120,0.01)$ in Fig.~\ref{fig:algo1} (left) and $(M,\sigma^{2})=(1000,0.01)$ in Fig.~\ref{fig:algo1} (right).
As expected from Theorem~\ref{theo:K1T}, 
larger $M$ and graph filters with sharper cutoffs (controlled by $\tau$) reduce the error rate; while larger $N$ raises the error rate.
Compared to the benchmark scheme based on SpecTemp \cite{segarra2017network} and checking Definition~\ref{def:lpf}, \Cref{algo:nec} achieves a lower MD rate at the same number of samples, and is less sensitive to graph size.


\begin{figure}[t]
\centering
\resizebox{0.99\linewidth}{!}{\sf \begin{tikzpicture}
\definecolor{green01270}{RGB}{0,127,0}

\definecolor{mycolor1}{rgb}{0.00000,0.44700,0.74100}%

\definecolor{darkgray176}{RGB}{176,176,176}
\definecolor{darkturquoise0191191}{RGB}{0,191,191}
\definecolor{lightgray204}{RGB}{204,204,204}

\begin{groupplot}[group style={group name=myplot,group size=2 by 1, horizontal sep= 2cm}]
\nextgroupplot[
legend cell align={left},
legend style={fill opacity=0.8, draw opacity=1, text opacity=1, draw=lightgray204}, font = \Large,
log basis x={10},
tick align=outside,
tick pos=left,
x grid style={darkgray176},
xlabel={\Large No.~of samples $M$},
xmajorgrids,
xmin=7.94328234724282, xmax=1258.92541179417,
xmode=log,
xtick style={color=black},
y grid style={darkgray176},
ylabel={\Large MD/FA Rate},
ymajorgrids,
ymin=-0.05, ymax=1.05,
ytick style={color=black},
width=7cm, 
height=7.5cm
]
\addplot [ultra thick, mycolor1, mark=o, mark size=4,  mark options={solid,rotate=180}]
table {%
10 0.95
20 0.91
50 0.42
100 0.27
200 0.02
500 0
1000 0
};
\addplot [ultra thick, mycolor1, dashed, mark=o, mark size=4,  mark options={solid,rotate=180}]
table {%
10 0
20 0
50 0
100 0
200 0
500 0
1000 0
};
\addplot [ultra thick, green!50!black, mark=triangle, mark size=4,  mark options={solid,rotate=180}  ]
table {%
10 0
20 0
50 0
100 0
200 0
500 0
1000 0
};
\addplot [ultra thick, green!50!black, dashed, mark=triangle, mark size=4,  mark options={solid,rotate=180} ]
table {%
10 0
20 0
50 0
100 0
200 0
500 0
1000 0
};

\addplot [ultra thick, red!50, mark=square, mark size=4, mark options={solid}]
table {%
10	1
20	1
50	1
100	0.91
200	0.74
500	0.35
1000	0.11
};
\addplot [ultra thick, red!50, dashed, mark=square, mark size=4, mark options={solid}]
table {%
10 0
20 0
50 0
100 0
200 0
500 0
1000 0
};

\coordinate (top) at (rel axis cs:0,1);

\nextgroupplot[
legend cell align={left},
legend style={
  fill opacity=0.8,
  draw opacity=1,
  text opacity=1,
  at={(0.03,0.97)},
  font = \Large,
  anchor=north west,
  draw=lightgray204
},
log basis x={10},
tick align=outside,
tick pos=left,
font = \Large,
x grid style={darkgray176},
xlabel={\Large Graph Size $N$},
xmajorgrids,
xmin=8.22340159426889, xmax=608.020895329329,
xmode=log,
xtick style={color=black},
y grid style={darkgray176},
ymajorgrids,
ymin = -0.01, ymax=0.4,
ytick style={color=black},
width=7cm, 
height=7.5cm
]
\addplot [ultra thick, mycolor1, mark=o, mark size=4,  mark options={solid, rotate=180}]
table {%
10 0
20 0
50 0
120 0
250 0
500 0.02
};\label{plots:1MD}
\addplot [ultra thick, mycolor1, dashed, mark=o, mark size=4,  mark options={solid, rotate=180}]
table {%
10 0
20 0
50 0
120 0
250 0
500 0
};\label{plots:1FA}
\addplot [ultra thick, green!50!black,mark=triangle, mark size=4,  mark options={solid, rotate=180} ]
table {%
10 0
20 0
50 0
120 0
250 0
500 0
};\label{plots:5MD}
\addplot [ultra thick, green!50!black, dashed, mark=triangle, mark size=4,  mark options={solid, rotate=180}]
table {%
10 0
20 0
50 0
120 0
250 0
500 0
};\label{plots:5FA}

\addplot [ultra thick, red!50, mark=square, mark size=4, mark options={solid}]
table {%
10	0.066
20	0.038
50	0.036
120	0.104
250	0.358
500	0.362
};\label{plots:GSOMD}

\addplot [ultra thick, red!50, dashed, mark=square, mark size=4, mark options={solid}]
table {%
10 0
20 0
50 0
120 0
250 0
500 0
};\label{plots:GSOFA}

\coordinate (bot) at (rel axis cs:1,0);

\end{groupplot}
\path (top|-current bounding box.north)--
      coordinate(legendpos)
      (bot|-current bounding box.north);
\matrix[
    matrix of nodes,
    anchor=south,
    draw,
    inner sep=0.2em,
    draw,
    font = \Large,
  ]at([yshift=1.5ex]legendpos)
  {\ref{plots:1MD}& $J=1$ MD&[5pt]
\ref{plots:5MD}& $J=5$ MD&[5pt]
\ref{plots:GSOMD}& SpecTemp MD&[5pt]\\
\ref{plots:1FA}& $J=1$ FA&[5pt]
\ref{plots:5FA}& $J=5$ FA&[5pt]
\ref{plots:GSOFA}& SpecTemp FA\\};

\end{tikzpicture}}\vspace{-.2cm}
\caption{\textbf{Error performance for detecting $1$-low pass graph signals} against (Left) sample size $M$, (Right) graph size $N$. We set $J=1$ for the experiment with SpecTemp \cite{segarra2017network}.}\vspace{-.2cm}
\label{fig:algo1}
\end{figure}

The second experiment considers the case with $K \geq 2$ and the filter parameter is fixed at $\tau = 0.6$. As the threshold $\delta$ in Algorithm~\ref{algo:nec} trades off between the false alarm and missed detection rates, we measure the detection performance through the AUROC score.
Fig.~\ref{fig:algo2} (top) show the performance of Algorithm~\ref{algo:nec} against $M$ while setting $(N,\sigma^{2})=(120,0.01)$, and $N$ while setting $(M,\sigma^{2})=( {0.01 N^2}  ,0.01)$. Note that in the latter case, we set $M = \Theta(N^2)$ to compensate for the increase of signal dimension [cf.~\eqref{eq:sample_complexity_K}].  
Again as predicted by Theorem~\ref{theo:Km_sample}, we observe that the detection performance improves ($\text{AUROC} \to 1$) as $M$, $N$ increase. 
{Additionally, we consider a case when $K$ is unknown and is estimated using the heuristic in Remark~\ref{remark:K}. Observe that the performance has only dropped slightly.}
Compared to the benchmark scheme based on SpecTemp \cite{segarra2017network}, we observe improved performance across different sample and graph sizes for the proposed \Cref{algo:nec}.

{The third experiment considers the sensitivity of \Cref{algo:nec} to the filter's sharpness and the modularity of the graphs. We fix $\sigma^2 = 0.01$, $N=120$, $M=50$, and $r = \log N / N$. Fig.~\ref{fig:algo2} (bottom)
plots the AUROC performance of Algorithm~\ref{algo:nec} against the filter's parameter $\tau$ while setting $p=4 \log N/N$, and against the graph's modularity $p$ while setting $\tau = 0.6$.
Observe the detection performance of Algorithm~\ref{algo:nec} improves as $\tau$ (filter's sharpness) increases and $p$ (modularity) increases, confirming the analysis in Sec.~\ref{sec:ana}.
}

\begin{figure}[t]
\centering
\resizebox{0.95\linewidth}{!}{\sf \begin{tikzpicture}

\definecolor{mycolor1}{rgb}{0.00000,0.44700,0.74100}%

\definecolor{darkgray176}{RGB}{176,176,176}
\definecolor{darkturquoise0191191}{RGB}{0,191,191}
\definecolor{darkviolet1910191}{RGB}{191,0,191}
\definecolor{lightgray204}{RGB}{204,204,204}
\definecolor{green01270}{RGB}{0,127,0}
\definecolor{mycolor3}{rgb}{0.92900,0.69400,0.12500}%
\begin{groupplot}[group style={group name=myplot,group size=2 by 1}]
\nextgroupplot[
legend cell align={left},
legend style={
  fill opacity=0.8,
  draw opacity=1,
  text opacity=1,
  at={(0.97,0.05)},
  anchor=south west,
  draw=lightgray204
},
log basis x={10},
tick align=outside,
tick pos=left,
x grid style={darkgray176},
font = \Large,
xlabel={\Large No. of samples $M$},
xmajorgrids,
xmin=7.94328234724282, xmax=1258.92541179417,
xmode=log,
xtick style={color=black},
y grid style={darkgray176},
ylabel={{\Large AUROC}},
ymajorgrids,
ymin=0.45, ymax=1.03,
ytick style={color=black},
width=7.5cm, 
height=8.5cm
]
\addplot [ultra thick, mycolor1, mark=diamond, mark size=4, mark options={solid}]
table {%
10 0.828
20 0.9355
50 0.9902
100 0.9971
200 1
500 1
1000 1
};
\addplot [ultra thick, mycolor1, dashed, mark=o, mark size=4, mark options={solid}]
table {%
10 0.842
20 0.92
50 0.9835
100 1
200 1
500 1
1000 1
};

\addplot [ultra thick, green!50!black, mark=triangle, mark size=4, mark options={solid,rotate=180}]
table {%
10 1
20 1
50 1
100 1
200 1
500 1
1000 1
};

\addplot [ultra thick, green!50!black, dashed, mark=triangle, mark size=4,mark options={solid, rotate=180} ]
table {%
10 1
20 1
50 1
100 1
200 1
500 1
1000 1
};

\addplot [ultra thick, red!50, mark=square, mark size=4, mark options={solid}]
table {%
10  0.5
20	0.5
50	0.5
100	0.5
200	0.505
500	0.58
1000  0.77
};

\addplot [ultra thick, darkturquoise0191191, mark=diamond, mark size=4,  mark options={solid}]
table {%
10 0.6589
20 0.7948
50 0.912
100 0.9888
200 1
500 1
1000 1
};

\addplot [ultra thick, darkturquoise0191191, dashed, mark=o, mark size=4,  mark options={solid}]
table {%
10 0.657
20 0.7357
50 0.8807
100 0.9564
200 0.9982
500 1
1000 1
};



\coordinate (top) at (rel axis cs:0,1);

\nextgroupplot[
legend cell align={left},
legend style={
  fill opacity=0.8,
  draw opacity=1,
  text opacity=1,
  at={(0.05,0.3)},
  anchor=south west,
  draw=lightgray204
},
log basis x={10},
tick align=outside,
tick pos=left,
x grid style={darkgray176},
xlabel={\Large Graph Size $N$},
xmajorgrids,
xmin=28, xmax=560,
xmode=log,
font = \Large,
xtick style={color=black},
y grid style={darkgray176},
ymajorgrids,
ymin=0.45, ymax=1.03,
ytick style={color=black},
width=7.5cm, 
height=8.5cm
]
\addplot [ultra thick, mycolor1, mark=diamond, mark size=4, ]
table {%
32 0.9502
52 0.9834
84 0.9894
120 1
240 1
480 1
};\label{plots:12}

\addplot [ultra thick, mycolor1, dashed, mark=o, mark size=4, mark options={solid}]
table {%
32 0.931
52 0.977
84 0.999
120 1
240 1
480 1
};\label{plots:14}

\addplot [ultra thick, green!50!black, mark=triangle, mark size=4, mark options={rotate=180}]
table {%
32 1
52 1
84 1
120 1
240 1
480 1
};\label{plots:52}

\addplot [ultra thick, green!50!black, dashed, mark=triangle, mark size=4, mark options={solid, rotate=180}]
table {%
32 1
52 1
84 1
120 1
240 1
480 1
};\label{plots:54}

\addplot [ultra thick, red!50, mark=square, mark size=4,  mark options={solid}]
table {%
32	0.545
52	0.535
84	0.505
120	0.5
240	0.5
480	0.83
};\label{plots:SpecTemp}

\addplot [ultra thick, darkturquoise0191191,  mark=diamond, mark size=4,  mark options={solid}]
table {%
32 0.8649
52 0.9209
84 0.9602
120 0.9986
240 1
480 1
};\label{plots:12est}

\addplot [ultra thick, darkturquoise0191191, dashed, mark=o, mark size=4,  mark options={solid}]
table {%
32 0.8748
52 0.8646
84 0.9066
120 0.9779
240 1
480 1
};\label{plots:14est}


\coordinate (bot) at (rel axis cs:1,0);

\end{groupplot}

\path (top|-current bounding box.north)--
      coordinate(legendpos)
      (bot|-current bounding box.north);
\matrix[ 
    matrix of nodes,
    anchor=south,
    draw,
    inner sep=0.2em,
    draw,
    font = \large,
  ]at([yshift=1.5ex]legendpos)
  {\ref{plots:12}& $(J,K)=(1,2)$&[2pt]
\ref{plots:14}& $(J,K)=(1,4)$&[2pt]
\ref{plots:52}& $(J,K)=(5,2)$&[2pt]\\
\ref{plots:12est}& $(J,K)=(1,2),{\hat K}_{\sf est}$&[2pt]
\ref{plots:14est}& $(J,K)=(1,4),{\hat K}_{\sf est}$&[2pt]
\ref{plots:54}& $(J,K)=(5,2), {\hat K}_{\sf est}$ &[2pt]\\
\ref{plots:SpecTemp}& SpecTemp&[2pt]
\\};

\end{tikzpicture}}\vspace{.1cm}

\resizebox{0.95\linewidth}{!}{\sf 

\begin{tikzpicture}
\definecolor{green01270}{RGB}{0,127,0}

\definecolor{mycolor1}{rgb}{0.00000,0.44700,0.74100}%

\definecolor{darkgray176}{RGB}{176,176,176}
\definecolor{darkturquoise0191191}{RGB}{0,191,191}
\definecolor{lightgray204}{RGB}{204,204,204}


\begin{groupplot}[group style={group name=myplot,group size=2 by 1}]
\nextgroupplot[
legend cell align={left},
legend pos = north west,
legend style={
  fill opacity=0.8,
  draw opacity=1,
  text opacity=1, 
  font=\Large,
  draw=lightgray204
},
tick align=outside,
tick pos=left,
x grid style={darkgray176},
xlabel={\Large Filter Parameter $\tau$},
xmajorgrids,
font = \Large,
xmin=0.055, xmax=1.045,
xtick style={color=black},
y grid style={darkgray176},
ylabel={\Large AUROC},
ymajorgrids,
ymin=0.45, ymax=1.05,
ytick style={color=black},
width=7.5cm, 
height=8.5cm
]
\addplot [ultra thick, mycolor1, mark=diamond, mark size=4, ]
table {%
0.1 0.607504
0.2 0.698184
0.3 0.800412
0.4 0.877496
0.5 0.946876
0.6 0.979644
0.7 0.988976
0.8 0.997028
0.9 0.99878
1 1
};
\addplot [ultra thick, darkturquoise0191191,  mark=diamond, mark size=4]
table {%
0.1 0.50875
0.2 0.5712
0.3 0.696
0.4 0.7547
0.5 0.799925
0.6 0.87305
0.7 0.939725
0.8 0.9666
0.9 0.99395
1 0.995625
};

\addplot [ultra thick, red!50, mark=square, mark size=4]
  table{%
0.1	0.5
0.2	0.5
0.3	0.51
0.4	0.5
0.5	0.505
0.6	0.53
0.7	0.595
0.8	0.585
0.9	0.645
1	0.62
};

\coordinate (top) at (rel axis cs:0,1);

\nextgroupplot[
legend cell align={left},
legend style={
  fill opacity=0.8,
  draw opacity=1,
  text opacity=1,
  at={(0.03,0.97)},
  anchor=north west,
  font = \Large,
  draw=lightgray204
},
tick align=outside,
tick pos=left,
x grid style={darkgray176},
xlabel={\Large Modularity of Graph $\bar{p}$},
xmajorgrids,
xmin=-0.245, xmax=7.345,
font = \Large,
xtick style={color=black},
y grid style={darkgray176},
ymajorgrids,
ymin=0.45, ymax=1.05,
ytick style={color=black},
width=7.5cm, 
height=8.5cm
]
\addplot [ultra thick, mycolor1, mark=diamond, mark size=4, ]
table {%
0.1 0.841258
0.2 0.869978
0.5 0.902954
1 0.927029
2 0.952961
3 0.966068
4 0.970284
5 0.971058
6 0.976208
7 0.991249
};
\addplot [ultra thick, darkturquoise0191191,  mark=diamond, mark size=4]
table {%
0.1 0.566275
0.2 0.629964
0.5 0.678675
1 0.746354
2 0.81631
3 0.861424
4 0.880417
5 0.89944
6 0.904552
7 0.900563
};

\addplot [ultra thick, red!50, mark=square, mark size=4]
  table{%
0.1	0.5
0.2	0.5
0.5	0.5
1	0.505
2	0.505
3	0.505
4	0.505
5	0.525
6	0.52
7	0.52
};

\end{groupplot}

\end{tikzpicture}}\vspace{-.1cm}
\caption{\textbf{AUROC performance} against (Top-Left) sample size $M$; (Top-Right) graph size $N$, 
where we set $J=1,K=2$ for the experiments with SpecTemp \cite{segarra2017network};
(Bottom-Left) filter's parameter $\tau$; (Bottom-Right) SBM parameter $\bar{p}$ such that $p=\bar{p}\log N/N$.
{In the above, $\hat{K}_{\sf est}$ refers to the heuristics in Remark~\ref{remark:K} which  estimates $K$ from data.}
}
\vspace{-.1cm}
\label{fig:algo2}
\end{figure}


Finally, we verify the robustness of Algorithm~\ref{algo:nec} to the graph topology on an actual network that is \emph{not generated by the SBM}, and detecing via graph signals that are not excited by \emph{white} input graph signal. This setup violates some of the assumptions required by our analysis. We consider the {\tt Political Books} network [available: \url{http://www.orgnet.com/}] with $N=105$ nodes and $|E| = 441$ edges, as illustrated in Fig.~\ref{fig:political_book} (left). The graph has roughly $K=2$ clusters.  

\begin{figure}[t]
\centering
\subfigure{\resizebox{0.495\linewidth}{!}{\sf \includegraphics{fig/fig7-1.png}}
}%
\subfigure{\resizebox{0.5\linewidth}{!}{\sf 



\begin{tikzpicture}

\definecolor{mycolor1}{rgb}{0.00000,0.44700,0.74100}%
\definecolor{darkgray176}{RGB}{176,176,176}
\definecolor{darkturquoise0191191}{RGB}{0,191,191}
\definecolor{green01270}{RGB}{0,127,0}
\definecolor{lightgray204}{RGB}{204,204,204}

\begin{axis}[
legend cell align={left},
legend style={
  fill opacity=0.8,
  draw opacity=1,
  text opacity=1,
  at={(0.97,0.03)},
  anchor=south east,
  draw=lightgray204
},
log basis x={10},
tick align=outside,
tick pos=left,
x grid style={darkgray176},
xlabel={No.~of samples $M$},
xmajorgrids,
xmin=7.94328234724282, xmax=1258.92541179417,
xmode=log,
xtick style={color=black},
y grid style={darkgray176},
ylabel={AUROC},
ymajorgrids,
ymin=0.719545, ymax=1.013355,
ytick style={color=black},
height=6cm,
width=5.5cm
]
\addplot [ultra thick, mycolor1, mark=triangle*, mark size=3, mark options={rotate=180}]
table {%
10 0.7329
20 0.7946
50 0.8924
100 0.9435
200 0.9891
500 0.9903
1000 1
};
\addlegendentry{$(R,\tau)=(10,0.2)$}
\addplot [ultra thick, mycolor1, dashed, mark=triangle*, mark size=3, mark options={rotate=180}]
table {%
10 0.8035
20 0.8937
50 0.9531
100 0.996
200 0.9996
500 0.9994
1000 1
};
\addlegendentry{$(R,\tau)=(25,0.2)$}
\addplot [ultra thick, green01270, mark=o, mark size=3]
table {%
10 0.9058
20 0.9475
50 0.9945
100 1
200 1
500 1
1000 1
};
\addlegendentry{$(R,\tau)=(10,0.5)$}
\addplot [ultra thick, green01270, dashed, mark=o,mark size=3]
table {%
10 0.9578
20 0.9851
50 1
100 1
200 1
500 1
1000 1
};
\addlegendentry{$(R,\tau)=(25,0.5)$}
\end{axis}

\end{tikzpicture}}
}\vspace{-.2cm}
\caption{\textbf{Political Books network.} (Left) The graph topology with $K=2$ clusters. (Right) AUROC performance against $M$ with \Cref{algo:nec}.
}\vspace{-.2cm}
\label{fig:political_book}
\end{figure}

We consider $R \in \{10,25\}$ sources simulated as exogeneous sources of excitation to the graph. Accordingly, the input signals to the graph filter is generated as ${\bm x}_m=\B {\bm z}_m, {\bm z}_m\in \mathbb{R}^R\sim N(\mathbf{0}, \I)$, 
where $\B$ is a sparse bipartite graph with connectivity $2\log N/N$ between the $R$ sources and the existing $N$ nodes on the graph.
The low and high pass graph filters are respectively generated as $\HS = ({\bm I}+\tau_2\nL)^{-1}$ and $\HS = {\bm I}+\tau_2\nL$.
Fig.~\ref{fig:political_book} (right) shows the AUROC performance against the number of samples $M$ collected. Observe that AUROC$\to 1$ as $M, R$ increase. The result illustrates that despite not all assumptions in Theorem~\ref{theo:Km_sample} are satisfied, Algorithm~\ref{algo:nec} remains effective in tackling Problem~\ref{prob:det}.\vspace{-.2cm}

\subsection{Application Examples}\vspace{-.1cm}
This subsection illustrates numerical results from applying the proposed low pass detection algorithm to {\sf (I)} robustify graph topology learning from corrupted signals, {\sf (II)} detecting antagonistic behavior in social networks, {\sf (III)} detecting anomaly status in power systems.\vspace{.1cm}


\begin{figure}
    \centering
    \resizebox{0.6\linewidth}{!}{\sf\includegraphics{fig/fig8.png}}
    \caption{ {\textbf{Batch-mode and random-mode pollution}, where {\color{red}red} part indicates the positions of corrupted signals}}
    \vspace{-.5cm}
    \label{fig:data-gen}
\end{figure}

\begin{figure*}[t]
\resizebox{0.97\linewidth}{!}{\sf \input{fig/fig9}}
\vspace{-.25cm}
  \caption{\textbf{AUROC performance under polluted/corrupted data.}  AUROC (left) against sparsity of outlier $p_s$ under batch-mode pollution with sparse pollution;  {(Middle-left) against batch size $M_{\sf batch}$ under batch-mode pollution with sparse pollution at $p_s=0.5$.} (Middle-right) against size of corrupted batch $M_{\sf ba}$ under batch-mode pollution with \cite{wang2022distributionally}'s model. (Right) against sparsity of outlier $p_s$ under random-mode pollution. 
} \vspace{-.25cm} \label{fig:graph_auc}
\end{figure*}

\subsubsection{Robustifying Graph Topology Learning}
Consider application {\sf (I)} 
via pre-screening with \Cref{algo:nec}. We aim to evaluate the approach on synthetic graph models and signals. Particularly, we simulate a graph with $N=20$ nodes, $G \sim {\sf SBM}(2,N, \log(N)/N,2 \log(N)/N)$, and the graph filter is given by $e^{\nL}$. The number of samples and observation noise variance are $(M,\sigma^2)=(2000,0.01)$.
To describe ${\cal M}_{\sf pollut}$, we simulate signals that are corrupted in a batch-mode manner, or are corrupted uniformly at random. For the batch-mode setting, among the $M$ samples, we randomly select ten time indices $\{t_1,\ldots,t_{10}\}$ as the starting time for signal corruption. Subsequent signals are randomly contaminated in batches with the duration $M_{\tt ba} \geq 1$, as illustrated by   {Fig.~\ref{fig:data-gen} (top)}. In other words, we have ${\cal M}_{\sf pollut} = \cup_{i=1}^{10} \{ t_i, \ldots, t_i+M_{\tt ba}-1 \}$. The uniform corruption setting is similar to the above where ${\cal M}_{\sf pollut}$ is selected by randomly picking $10\%$ of the indices from $\{1,\ldots,2000\}$, as illustrated by   {Fig.~\ref{fig:data-gen} (bottom)}. 

We evaluate the graph topology learnt using GL-SigRep or SpecTemp combined with the proposed pre-screening scheme by \Cref{algo:nec} (denoted `{\sf Pre-screen GL-SigRep/SpecTemp}'). As benchmarks, we also compare with {\sf Inpaint} model \cite{berger2020efficient}, {\sf OR-GL} \cite{araghi2023outlier}, {\sf LS-PGD} \cite{wang2022distributionally}, the plain {\sf GL-SigRep} \cite{dong2016learning} and {\sf SpecTemp} \cite{segarra2017network}. Note that the last two algorithms are not designed for graph topology learning with corrupted signals.
The performance of graph learning is assessed via AUROC score. 
For the pre-screening procedure, we divide the graph signal dataset into batches of $m_{\sf batch} = 50$ samples, and apply \Cref{algo:nec} with the threshold $\delta$. 
The {\sf Inpaint} model \cite{berger2020efficient} and the {\sf OR-GL} \cite{araghi2023outlier} underwent 20 Monte-Carlo trials each, whereas the remaining simulations were conducted with at least 200 Monte-Carlo trials.


We test the graph learning performance under batch-mode corruption. Fig.~\ref{fig:graph_auc}  {(left)} shows the graph learning performance against the density of sparse noise $p_s$ while fixing the contamination duration $M_{\sf ba} = 20$.
For each $m \in {\cal M}_{\sf pollut}$, the corrupted signal  {${\bm p}_m = \HS {\bm x}_m + \Delta {\bm p}_m$} is generated by contaminating the low pass signal $\HS {\bm x}_m$ with missed observations and sparse outlier noise,  {the latter is modeled using a sparse vector $\Delta {\bm p}_m$} -- $10\%$ of the entries will be missed and a $p_s N$-sparse noise will be added with uniformly selected coordinates. For the sparse noise, its non-zero entries are sampled from ${\cal N}( 3 \delta_y, 0.2 \delta_y )$ with $\delta_y = \max_{m} \| \HS {\bm x}_m \|_\infty$. The latter model is designed such that the contaminated graph signals can evade detection by simple schemes such as thresholding on  magnitudes.
Moreover, \Cref{algo:nec} with $K=2$ is applied with $\delta = 0.8$. 
In the figure, we observe that the pre-screened GL-SigRep scheme achieves competitive graph learning performance especially when $p_s$ is large. The performance improvement is significant compared to original GL-SigRep. The pre-screening procedure robustifies graph learning against outliers corruption \eqref{eq:y-pollute}. 
 {Additionally, Fig.~\ref{fig:graph_auc} (mid-left) follows the same simulation setting as Fig.~\ref{fig:graph_auc} (left) with $p_s=0.5$. This figure demonstrates a performance tradeoff when deciding the batch size $m_{\sf batch}$ of the pre-screening procedure. We observe that the optimal $m_{\sf batch}$ is around $20$-$80$ which is in the same order as $M_{\sf ba} = 20$.
}

The second example considers batch-mode corruption with the outlier signal model taken from LS-PGD in \cite[Eq.~(5)]{wang2022distributionally}. Here, the outlier model includes both \emph{uncertainty and noise}: for any $m \in {\cal M}_{\sf pollut}$, we have ${\bm y}_m \sim {\cal N}( \bm{\mu}^\star, \GSO^\dagger + {\bm I} )$ where $\bm{\mu}^\star \sim {\cal N}( {\bm 0}, {\bm I} )$. Fig.~\ref{fig:graph_auc}  {(mid-right)} compares the graph learning performance against the corruption duration $M_{\sf ba}$, where we applied the pre-screened GL-SigRep scheme with the threshold $\delta = 0.6$. Observe that the pre-screened scheme delivers favorable performance across the tested range of $M_{\sf ba}$.

Lastly, Fig.~\ref{fig:graph_auc}  {(right)} compares the performance of graph learning under the uniform corruption setting with a 10\% contamination rate. The other simulation setting follows from Fig.~\ref{fig:graph_auc}  {(left)}. Note that this setting favors {\sf OR-GL} \cite{araghi2023outlier}. We observe that the pre-screened GL-SigRep scheme performs worse, yet it still outperforms the other benchmarks such as directly applying GL-SigRep on the corrupted data.

\subsubsection{Detecting Antagonistic Ties in Opinion Dynamics}
{As discussed in Sec.~\ref{sec:ex}, opinion data tend to appear as low pass (resp.~high pass) when the social network is dominated by friendly (resp.~antagonistic) ties. This observation inspires us to use the normalized score ${\sf score}_{\sf Alg1} := \frac{1}{K} \Kmeans( \widehat{\bm U}_K )$ in \Cref{algo:nec} to measure the strength of antagonistic ties. As benchmarks, we compare two signed graph learning methods: GLASSO in \cite{friedman2008sparse}, and the method in \cite{matz2020learning}. The two methods learn the weighted signed graph. To quantify the strength of antagonistic ties, we consider the normalized score 
\beq \label{eq:score}
{\sf score}_{\sf GL} = \frac{ \| {\bm A}_{E^{-}} \|_{\fro} }{ \| |{\bm A}_{E^+}| + |{\bm A}_{E^-} | \|_{\fro}}, 
\eeq
such that ${\bm A}_{E^{-}}$, ${\bm A}_{E^{+}}$ denote the adjacency matrix of negative, positive edges learnt, respectively. For ease of comparison, these scores are normalized to the range $[0,1]$, where a smaller (larger) value suggests more friendly (antagonistic) ties\footnote{We acknowledge that as the goal for \Cref{algo:nec} is different from \cite{friedman2008sparse, matz2020learning} by nature, it is impossible to make a completely fair comparison. Our examples serve as a reference to demonstrate that the proposed detector for antagonistic ties produce reasonable result with little computation overhead.}.}


\begin{figure}
\centering
\resizebox{0.99\linewidth}{!}{\sf\definecolor{darkgray176}{RGB}{176,176,176}
\definecolor{darkturquoise0191191}{RGB}{0,191,191}
\definecolor{green01270}{RGB}{0,127,0}
\definecolor{lightgray204}{RGB}{204,204,204}
\definecolor{mycolor2}{rgb}{0.00000,0.44700,0.74100}%
\definecolor{mycolor3}{rgb}{0.85000,0.32500,0.09800}%
\definecolor{mycolor1}{rgb}{0.92900,0.69400,0.12500}%
\definecolor{mycolor4}{rgb}{0.49400,0.18400,0.55600}%
\definecolor{mycolor5}{rgb}{0.46600,0.67400,0.18800}%
%
%
\begin{tikzpicture}

\begin{groupplot}[group style={group name=myplot,group size=2 by 1}]
\nextgroupplot[
legend cell align={left},
legend pos = north west,
legend style={
  fill opacity=0.8,
  draw opacity=1,
  text opacity=1,
   font=\Large,
  draw=lightgray204
},
tick align=outside,
tick pos=left,
x grid style={darkgray176},
xlabel={\LARGE $|E^{-}| / |E|$},
xmajorgrids,
xmin=0,
xmax=0.6,
xtick style={color=black},
 y grid style={darkgray176},
ylabel={\Large Normalized Scores},
ymajorgrids,
ymin=0,
ymax=0.8,
ytick style={color=black},
width=6.5cm, 
height=7.5cm
]
\addplot [color=mycolor2, line width=1.5pt, mark= *, mark size = 4, mark options={solid, mycolor2}]
  table[row sep=crcr]{%
0	0.111555603537955\\
0.05	0.156599325718283\\
0.1	0.198654956314042\\
0.2	0.33190395082266\\
0.3	0.525660694532152\\
0.4	0.640996049443781\\
0.6	0.666098928768766\\
};

\addplot [densely dashed, color=green!50!black, line width=1.5pt, mark=o, mark size = 4, mark options={solid, green!50!black}]
  table[row sep=crcr]{%
0	4.22526620908383e-10\\
0.05	0.0399060554586674\\
0.1	0.136374966511475\\
0.2	0.395267819815029\\
0.3	0.549593864307343\\
0.4	0.643160215631223\\
0.6	0.731611465146178\\
};

\addplot [densely dashed, color=red,  line width=1.5pt, mark=square, mark size = 4, mark options={solid, red}]
  table[row sep=crcr]{%
0	0.00944562097369884\\
0.05	0.0428708722070309\\
0.1	0.113948166161689\\
0.2	0.325055216556284\\
0.3	0.515743229761056\\
0.4	0.633533747072444\\
0.6	0.749042765601714\\
};

\coordinate (top) at (rel axis cs:0,1);

\nextgroupplot[%
legend style={yshift=-2.5cm},
xshift = 1cm,
xmajorgrids,
xmin=0,
xmax=0.6,
xtick style={color=black},
xlabel style={font=\color{white!15!black}},
tick align=outside,
tick pos=left,
x grid style={darkgray176},
y grid style={darkgray176},
ytick style={color=black},
xlabel={\LARGE $|E^{-}| / |E|$},
ymode=log,
ymin=1e-3,
ymax=1e3,
ylabel style={font=\color{white!15!black}},
ylabel={\Large Avg.~Running Time(s)},
xmajorgrids,
ymajorgrids,
width=6.5cm, 
height=7.5cm
]

\addplot [color=mycolor2, line width=1.5pt, mark=*, mark size = 4, mark options={solid, mycolor2}]
  table[row sep=crcr]{%
0	0.01681885998\\
0.05	0.00448990422\\
0.1	0.00419783076\\
0.2	0.00480204584\\
0.3	0.00493729578\\
0.4	0.0045071267\\
0.6	0.00421325662\\
};
\addlegendentry{Algorithm~\ref{algo:nec}}

\addplot [densely dashed, color=green!50!black, line width=1.5pt,  mark=o, mark size = 4, mark options={solid, green!50!black}]
  table[row sep=crcr]{%
0	269.20459339\\
0.05	123.51822178328\\
0.1	79.91764500998\\
0.2	46.5680457384\\
0.3	42.7516077551\\
0.4	42.30209809502\\
0.6	41.95416399078\\
};
\addlegendentry{\cite{friedman2008sparse}, $\rho_{\sf GL} = 10^3$}

\addplot [densely dashed, color=red, line width=1.5pt, mark=square, mark size = 4, mark options={solid, red}]
  table[row sep=crcr]{%
0	18.28328223502\\
0.05	18.04229696412\\
0.1	17.86748534\\
0.2	18.395091435\\
0.3	18.16909106752\\
0.4	18.13943361418\\
0.6	18.3208259517\\
};
\addlegendentry{\cite{matz2020learning}, $\mu_{\sf SGL}=10^3$}

\end{groupplot}
\end{tikzpicture}%
}\vspace{-0.2cm}
\caption{{\bf Synthetic Signed Graph.} (Left) Normalized scores of detecting antagonistic ties against $|E^{-}| / |E|$. (Right) Average runtime against $|E^{-}| / |E|$.}
\vspace{-.3cm}
\label{fig:signed4}
\end{figure}

{Our first experiment aims at verifying the above application of \Cref{algo:nec} on synthetic data. We set $G \sim {\sf SBM}(2,N, \log(N)/N,4 \log(N)/N)$ with $N=100$ nodes. The observation noise satisfies $\sigma^2 = 0.01$ and we consider generating  $M=200$ samples using the model in \eqref{eq:signedeq} with $\alpha=0.8$ and ${\bm B}={\bm I}$. For GLASSO, we pick the regularization parameter as $\rho_{\sf GL} = 1000$, and for \cite{matz2020learning}, we pick $\mu_{\sf SGL} = 1000$.
Fig.~\ref{fig:signed4} compares ${\sf score}_{\sf Alg1}, {\sf score}_{\sf GL}$ and the average runtime against the portion of negative edges in the ground truth graph where a portion of $|E^{-}| / |E|$ edges are randomly flipped to negative, with $50$ Monte-Carlo trials. All algorithms detect the increased portion of antagonistic ties in the ground truth signed graphs, as indicated by the increasing detection scores as $| E^- | / |E| \to 1$. \Cref{algo:nec} has a significantly lower runtime than other benchmarks. 
}



Next, we apply \Cref{algo:nec} to a US Senate rollcall dataset. The dataset is taken from the $117$th US Congress [available: \url{https://voteview.com}], recording $M = 949$ rollcalls of votes made by $N=97$ members. The $M$ rollcalls are divided into 4 groups based on the attribute `vote question' as: ``On the Nomination", ``On the Cloture Motion", ``On the Amendment" and others. Each group is exemplified by rollcalls of different nature, as shown in Table~\ref{fig:fscore}. We postulate that the Senators' networks exhibit different levels of antagonistic ties in each of the group. For example, the opinion formation process on nomination of government positions (``On the Nomination'') may display more distrusts as opposed to the process on modifying a bill (``On the Amendment'').  

\colorlet{tc11}{orange!37}
\colorlet{tc12}{orange!39}
\colorlet{tc13}{orange!2}
\colorlet{tc14}{orange!6}
\colorlet{tc21}{orange!14} 
\colorlet{tc22}{orange!18}
\colorlet{tc23}{orange!5}
\colorlet{tc24}{orange!10}
\colorlet{tc31}{orange!20} 
\colorlet{tc32}{orange!25}
\colorlet{tc33}{orange!0}
\colorlet{tc34}{orange!6}
\colorlet{tc41}{orange!5} 
\colorlet{tc42}{orange!1}
\colorlet{tc43}{orange!0}
\colorlet{tc44}{orange!0}
\colorlet{tc51}{orange!17} 
\colorlet{tc52}{orange!22}
\colorlet{tc53}{orange!0}
\colorlet{tc54}{orange!7}
\colorlet{tc61}{orange!29} 
\colorlet{tc62}{orange!37}
\colorlet{tc63}{orange!1}
\colorlet{tc64}{orange!18}
\begin{table}[t]
\centering
\fbox{
\parbox[b]{0.45\linewidth}
{\tiny
\textbf{On the Nomination (N)}: 
\emph{e.g.,} ``Thomas J. Vilsack, of Iowa, to be Secretary of Agriculture", ``Rahm Emanuel, of Illinois, to be Ambassador to Japan", ...\vspace{0.1cm}\\
\textbf{On the Cloture Motion (C)}: 
\emph{e.g.,} ``Beth Robinson, of Vermont, to be United States Circuit Judge for the Second Circuit", ``Douglas R. Bush, of Virginia, to be an Assistant Secretary of the Army",  ...\vspace{0.1cm}\\
\textbf{On the Amendment (A)}: 
\emph{e.g.,} ''To establish a deficit-neutral reserve fund relating to COVID-19 vaccine administration and a public awareness campaign", ``In the nature of a substitute", ``To improve the bill", ...\vspace{0.1cm}\\
\textbf{Others (O)}: 
\emph{e.g.,} ``A bill to provide for reconciliation pursuant to title II of S. Con. Res. 5", ``A resolution impeaching Donald John Trump, President of the United States, for high crimes and misdemeanors", ...
}
}~~\resizebox{0.48\linewidth}{!}{\sf \raisebox{3.3cm}
{
{\large
\begin{tabular}{c|c|c|c|c}
\toprule
 & \textbf{N} & \textbf{C} & \textbf{A} & \textbf{O} \\
\midrule
\bfseries \Cref{algo:nec} &  \cellcolor{tc11}  & \cellcolor{tc12} & \cellcolor{tc13} & \cellcolor{tc14} \\ 
\bfseries (0.08 sec.) &  \multirow{-2}{*}{\cellcolor{tc11}\textbf{0.37}} & \multirow{-2}{*}{\cellcolor{tc12}\textbf{0.39}} & \multirow{-2}{*}{\cellcolor{tc13}\textbf{0.02}} & \multirow{-2}{*}{\cellcolor{tc14}\textbf{0.06}} \\  
\midrule
\cite{friedman2008sparse}, $\rho_{\sf GL} = 10$ & \cellcolor{tc21} & \cellcolor{tc22} & \cellcolor{tc23} & \cellcolor{tc24} \\ 
(3969.92~sec.) & \multirow{-2}{*}{\cellcolor{tc21}{0.14}} & \multirow{-2}{*}{\cellcolor{tc22}{0.18}} & \multirow{-2}{*}{\cellcolor{tc23}{0.05}} & \multirow{-2}{*}{\cellcolor{tc24}{0.10}}  \\  
\midrule
\cite{friedman2008sparse}, $\rho_{\sf GL} = 150$ & \cellcolor{tc31} & \cellcolor{tc32} & \cellcolor{tc33} & \cellcolor{tc34} \\ 
(586.48~sec.) & \multirow{-2}{*}{\cellcolor{tc31}{0.20}} & \multirow{-2}{*}{\cellcolor{tc32}{0.25}} & \multirow{-2}{*}{\cellcolor{tc33}{0.0}} & \multirow{-2}{*}{\cellcolor{tc34}{0.06}}  \\    
\midrule
\cite{matz2020learning}, $\mu_{\sf SGL} = 300$ & \cellcolor{tc41} & \cellcolor{tc42} & \cellcolor{tc43} & \cellcolor{tc44} \\ 
(108.21~sec.) & \multirow{-2}{*}{\cellcolor{tc41}{0.05}} & \multirow{-2}{*}{\cellcolor{tc42}{0.01}} & \multirow{-2}{*}{\cellcolor{tc43}{0.0}} & \multirow{-2}{*}{\cellcolor{tc44}{0.0}}  \\  
\midrule
\cite{matz2020learning}, $\mu_{\sf SGL} = 1000$ & \cellcolor{tc51} & \cellcolor{tc52} & \cellcolor{tc53} & \cellcolor{tc54} \\ 
(124.28~sec.) & \multirow{-2}{*}{\cellcolor{tc51}{0.17}} & \multirow{-2}{*}{\cellcolor{tc52}{0.22}} & \multirow{-2}{*}{\cellcolor{tc53}{0.0}} & \multirow{-2}{*}{\cellcolor{tc54}{0.07}}  \\   
\midrule
\cite{matz2020learning}, $\mu_{\sf SGL} = 2000$ & \cellcolor{tc61} & \cellcolor{tc62} & \cellcolor{tc63} & \cellcolor{tc64} \\ 
(116.16~sec.) & \multirow{-2}{*}{\cellcolor{tc61}{0.29}} & \multirow{-2}{*}{\cellcolor{tc62}{0.37}} & \multirow{-2}{*}{\cellcolor{tc63}{0.01}} & \multirow{-2}{*}{\cellcolor{tc64}{0.18}}  \\   
\bottomrule
\end{tabular}
}}}
\vspace{0.1cm}
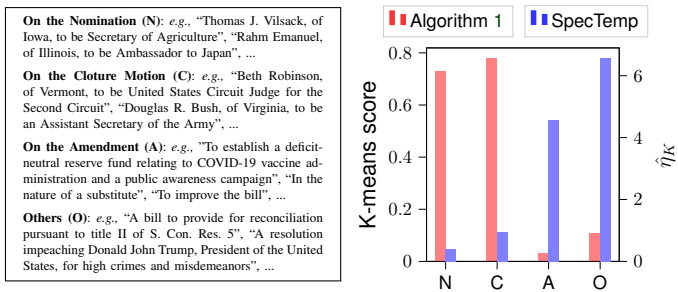
\caption{{\textbf{Detecting Antagonistic Ties in US Senate Dataset.} (Left) Examples of rollcall descriptions in each group. (Right) Normalized scores, ${\sf score}_{\sf Alg1}, {\sf score}_{\sf GL}$, computed from the rollcalls of each group (see below). The bracketed numbers are the total computation time for each algorithm.} 
}
\label{fig:fscore}
\end{table}

We process the data by assigning a score of $+1, 0, -1$ for a `Yay', `Abstention', `Nay' vote, respectively, and set the number of clusters to $K=2$ since there are two major parties. Table \ref{fig:fscore} shows the normalized scores \eqref{eq:score} computed from the 4 groups of rollcalls. 
Observe that for \Cref{algo:nec}, the normalized K-means scores are significantly higher for the group with ``On the Nomination'', ``On the Cloture Motion''. This indicates that the graph filter processes involved are \emph{non low pass} and shows traces of antagonistic ties. This is reasonable due to the common perception that these rollcalls are often contestable, even among the Senators of the same political party. On the other hand, the $K$-means score is lower for ``On the Amendment'', indicating the prevalence of \emph{low pass} processes as the Senators tend to reach consensus for these rollcalls. 
{Meanwhile, the graph learning methods are sensitive to the regularization parameters $\rho_{\sf GL}, \lambda_{\sf SGL}$.}


\begin{figure}[t]
\centering
\resizebox{0.45\linewidth}{!}{\sf \includegraphics{fig/fig11-1.pdf}
}~~
\resizebox{0.45\linewidth}{!}{\sf \begin{tikzpicture}

\definecolor{darkgray176}{RGB}{176,176,176}
\definecolor{darkturquoise0191191}{RGB}{0,191,191}
\definecolor{lightgray204}{RGB}{204,204,204}
\definecolor{mycolor2}{rgb}{0.00000,0.44700,0.74100}%
\begin{axis}[
legend cell align={left},
legend style={
  fill opacity=0.8,
  draw opacity=1,
  text opacity=1,
  font=\large,
  draw=lightgray204
},
legend pos = north east,
tick align=outside,
tick pos=left,
x grid style={darkgray176},
xlabel={\Large Trust parameter $\alpha$},
xmajorgrids,
xmin=0.005, xmax=0.9,
xtick style={color=black},
y grid style={darkgray176},
ylabel={\Large Normalized Scores},
ymajorgrids,
ymin=0, ymax=1.1,
ytick style={color=black},
width=6.5cm, 
height=7.5cm
]

\addplot [line width=1.5pt,mycolor2, mark=*, mark size=4, mark options={solid}]
table {%
0.05	0.471810197350928
0.1	0.477007635743769
0.2	0.476811241268496
0.3	0.472335311504193
0.4	0.474037729509744
0.5	0.478398146159754
0.6	0.480854958483187
0.7	0.478283035329752
0.8	0.483162211139103
0.9	0.48228988563693
};
\addlegendentry{Algorithm~\ref{algo:nec}}

\addplot [line width=1.5pt, green!50!black, mark=o, mark size=4, mark options={solid,fill opacity=0}]
  table{%
0.05	0.338676696442233
0.1	0.280376837720395
0.2	0.163806612687154
0.3	0.0959718737242619
0.4	0.02
0.5	0.01
0.6	0.01
0.7	0
0.8	0
0.9	0
};
\addlegendentry{\cite{friedman2008sparse}, $\rho_{\sf GL} = 5$}

\addplot [line width =1.5pt,red, mark=square, mark size=4, mark options={solid,fill opacity=0}]
  table{%
0.05	0.705695816547128
0.1	0.707342774204975
0.2	0.700018782074414
0.3	0.691789899220553
0.4	0.672465866779045
0.5	0.656167601640399
0.6	0.62311662222207
0.7	0.626504571395341
0.8	0.620489982091969
0.9	0.642139875943407
};
\addlegendentry{\cite{matz2020learning}, $\mu_{\sf SGL} = 0.5$}




\end{axis}

\end{tikzpicture}}\vspace{-.1cm}
\caption{\textbf{Highland tribes network.} (Left) The network: the red and grey lines are negative and positive edges, and the $R=5$ rectangular nodes are the excitation nodes selected. (Right) Score against trust parameter $\alpha \in (0,1)$ [cf.~\eqref{eq:opi_sign}].
}\vspace{-.3cm}
\label{fig:hightribe}
\end{figure}

{Lastly, we apply \Cref{algo:nec} to synthetic data generated from a real-world graph.
Consider the {\tt highland-tribes} graph [available: \url{http://konect.cc/networks/ucidata-gama/}] with $N+R=16$ agents and $|E^+| = 29$, $|E^-| = 29$ edges, from which we select $R=5$ agents as stubborn agents. We simulate the opinion dynamics by \eqref{eq:opi_sign} and generate $M=15$ samples with varying (dis)trust parameter $\alpha > 0$. The number of the clusters in the graph is estimated as $K=2$.
Fig.~\ref{fig:hightribe} (right) shows the normalized scores of benchmark algorithms against $\alpha$. Observe that the scores with Algorithm~\ref{algo:nec} and \cite{matz2020learning} are bounded away from zero consistently regardless of $\alpha$, indicating the detection of antagonistic ties. Meanwhile, GLASSO \cite{friedman2008sparse} fails to recognize the antagonistic ties when $\alpha > 0.4$. We speculate that this is due to the effects of low-rank excitation in the synthetic graph signals.}



\subsubsection{Detecting Anomalies in Power Systems}
Consider the voltage phasor data on an IEEE-118 bus test grid [available: \url{https://zenodo.org/record/5816149}]. 
The attack model is \eqref{eq:yesfdi} with sparse attack vector $\bm{\delta}_t$ such that $\big[ \bm{\delta}_t \big]_k =-A e^{j a_{t}}$ with the attack angle $a_t$ uniformly distributed in $[0,5^{\circ}]$.

Fig.~\ref{fig:fdia} shows the detection probability versus the attack magnitude $A \in [0,3]$ and different number of attacked buses $|{\cal D}_t|$. We compare \Cref{algo:nec} (assumed $K=2$) against the detection method in \cite{DrayerAC} which requires knowledge of the network topology. Although the power network is a man-made system, its topology may not be correctly estimated when energized \cite{cavraro2017voltage}. As such, we evaluate the performance when the graph topology is perturbed with $5\%$ of random edge connection/disconnection. On the other hand, \Cref{algo:nec} is a blind method that is robust to topology error. 
For the experiment, the detection probability is averaged over 100 trials and \Cref{algo:nec} is set to use 100 samples per detection. Our results demonstrate that \Cref{algo:nec} successfully detects FDIA events, and its performance is comparable to \cite{DrayerAC}, despite not relying on topology information. 

\begin{figure}[t]
\centering
\resizebox{0.8\linewidth}{!}{\sf 
%
\definecolor{darkgray176}{RGB}{176,176,176}
\definecolor{darkturquoise0191191}{RGB}{0,191,191}
\definecolor{green01270}{RGB}{0,127,0}
\definecolor{lightgray204}{RGB}{204,204,204}
\definecolor{mycolor2}{rgb}{0.00000,0.44700,0.74100}%
\definecolor{mycolor3}{rgb}{0.85000,0.32500,0.09800}%
\definecolor{mycolor1}{rgb}{0.92900,0.69400,0.12500}%
\definecolor{mycolor4}{rgb}{0.49400,0.18400,0.55600}%
\definecolor{mycolor5}{rgb}{0.46600,0.67400,0.18800}%
%
%
\begin{tikzpicture}

\begin{axis}[
legend cell align={left},
legend style={fill opacity=0.8, draw opacity=1, text opacity=1, draw=lightgray204,text depth=.1ex,nodes={scale=0.8, transform shape}, at={(1,0.45)},
  anchor=east,},
tick align=outside,
tick pos=left,
 x grid style={darkgray176},
xlabel={$A$},
xmajorgrids,
xmin=-0.05,
xmax=3.5,
xtick style={color=black},
 y grid style={darkgray176},
ylabel={Detection Probability},
ymajorgrids,
ymin=0,
ymax=1.05,
ytick style={color=black},
width = 9.5cm,
height = 6cm
]


\addplot [color=mycolor3, line width=2.0pt, mark=o, mark options={solid, mycolor3}]
  table[row sep=crcr]{%
0	0.02\\
0.1	0.02\\
0.2	0\\
0.5	0\\
1	0.59\\
2	0.88\\
3	0.9\\
};
\addlegendentry{$|{\cal D}_t| = 5$, Alg.1}

\addplot [color=mycolor2, line width=2.0pt, mark=o, mark options={solid, mycolor2}]
  table[row sep=crcr]{%
0	0.02\\
0.1	0\\
0.2	0\\
0.5	0.25\\
1	0.73\\
2	0.85\\
3	0.89\\
};
\addlegendentry{$|{\cal D}_t| = 2$, Alg.1}


\addplot [color=mycolor3, dashed, line width=2.0pt, mark=square, mark options={solid, mycolor3}]
  table[row sep=crcr]{%
0	0.0711111111111111\\
0.1	0.02\\
0.2	0.0122222222222222\\
0.5	0.425555555555556\\
1	0.987777777777778\\
2	1\\
3	1\\
};
\addlegendentry{$|{\cal D}_t| = 5$, \cite{DrayerAC}}

\addplot [color=mycolor2, dashed, line width=2.0pt, mark=square, mark options={solid, mycolor2}]
  table[row sep=crcr]{%
0	0.0711111111111111\\
0.1	0.0366666666666667\\
0.2	0.0255555555555556\\
0.5	0.18\\
1	0.825555555555556\\
2	0.985555555555556\\
3	1\\
};
\addlegendentry{$|{\cal D}_t| = 2$, \cite{DrayerAC}}


\addplot [color=mycolor3, dotted, line width=2.0pt, mark=triangle, mark options={solid, mycolor3}]
  table[row sep=crcr]{%
0	0.0301666666666667\\
0.1	0.00516666666666667\\
0.2	0.00555555555555555\\
0.5	0.00727777777777778\\
1	0.0213888888888889\\
2	0.0557222222222222\\
3	0.153222222222222\\
};
\addlegendentry{$|{\cal D}_t| = 5$, \cite{DrayerAC} w/ err}

\addplot [color=mycolor2, dotted, line width=2.0pt, mark=triangle, mark options={solid, mycolor2}]
  table[row sep=crcr]{%
0	0.0301666666666667\\
0.1	0.0140555555555556\\
0.2	0.0075\\
0.5	0.00661111111111111\\
1	0.0108333333333333\\
2	0.0283333333333333\\
3	0.0755555555555556\\
};
\addlegendentry{$|{\cal D}_t| = 2$, \cite{DrayerAC} w/ err}

\end{axis}
\end{tikzpicture}%
}\vspace{-.3cm}
\caption{\textbf{Detection probability versus FDI attack magnitude $A_t$.} We use `{\sf w/ err}' for \cite{DrayerAC} with perturbed topology, $|{\cal D}_t|$ is the no.~of attacked buses.
}\vspace{-.2cm}
\label{fig:fdia}
\end{figure}

\section{Conclusions}
This paper has initiated the study of a blind detection problem for low pass or smooth graph signals. The problem is motivated by the need to validate low-passness of graph signals before they can be used in GSP pipelines such as graph topology learning. We derive a detector that is inspired by the unique spectral pattern manifested by low pass graph signals defined on modular graphs and analyze their finite-sample complexity. Lastly, we discuss the applications of the proposed detectors on robustifying graph learning and anomaly detection in opinion dynamics data, power system data. Future works include extending the current method to directed graphs, studying applications to other domains and performance analysis on general modular graphs. 

\appendix
\subsection{Useful Lemmas} \label{app:Supple}
Below we state three useful results that will be instrumental to the subsequent analysis in this appendix.




\begin{Lemma}\cite[Lemma 3.1]{2011RoheSBM}
 For $G \sim {\sf SBM}( K,N, r,p)$, we have $\eA =\EE [ {\bm A} ]= {\bm Z} {\bm P} {\bm Z}^{\top}$, where ${\bm P} = p {\bf I}_K + r {\bf 1}_K {\bf 1}_K^\top$ and ${\bm Z} \in \{ 0, 1\}^{N \times K}$ is the cluster membership matrix. Moreover, $\fancyAn=\fancyD^{-1/2}\fancyA\fancyD^{-1/2}$ and $\fancyLn=\I-\fancyAn$, where
$\fancyD_{ii}=\sum_{j=1}^N \eA_{ij}$. The eigenvectors that correspond to the smallest $K$ eigenvalues of $\fancyLn$ are given by \vspace{-.0cm}
\beq \label{eq:fancyVKform}
\fancyV_K = \Z \left(\Z^{\top} \Z \right)^{-1/2} {\cal U},\vspace{-.0cm}
\eeq 
where ${\cal U} \in \mathbb{R}^{K\times K}$ is an orthogonal matrix.
\end{Lemma}

\begin{Lemma} \label{lem:Vpopu}
Let $ G \sim {\sf SBM}(K, N, r, p)$ with $p \geq r > 0$ and $\frac{p}{K}+r \geq \frac{32\log N+1}{N}$, let $\V_{K}$, $\fancyV_{K}$
denote the columns of the first $K$ 
eigenvectors of ${\nL}$, $\fancyLn$. Then,
with probability at least $1-2 / N$, there exists an orthogonal matrix $\fancyO_K \in \RR^{K \times K}$:\vspace{-.0cm}
\begin{equation} \notag
\begin{aligned}
\|\V_{K}-\fancyV_{K}\fancyO_{K}\|_{\fro} & \leq \sqrt{2} \|\sin \Theta(\V_{K}, \fancyV_{K})\|_{\fro} 
\leq\frac{35 \sqrt{K^3 \log N}}{ \sqrt{ p( N-K ) }}. 
\end{aligned}\vspace{-.0cm}
\end{equation}
\end{Lemma}
The proof is relegated to \Cref{app:vpop}. The conditions on $p,r,K$ pertain to the difference between inter- and intra-cluster connection probabilities.
The following lemma is obtained by combining \cite[Theorem 2.1]{2015Bunea} and Davis-Kahan theorem:
\begin{Lemma} \label{lem:sample}
If $\Delta_{K,K}> 0$, there exists an orthogonal matrix $\hatO_K$ such that with probability at least $1-5/M$,\vspace{-.0cm}
\begin{align}
& \frac{1}{\sqrt{2}} \| \barU_K \hatO_K - \hatU_K \|_{\fro} \leq \|\sin \Theta(\hatU_K, \barU_K)\|_{\fro} \leq 
\frac{2 \sqrt{K}}{ \eigB_K } \Emc \notag
\end{align} 
where ${\rm c}_1$ is an absolute constant independent of $N,M$ and
\begin{align} \label{eq:cov_conc_K}
\Emc :=  2 c_{1} \sqrt{ {\log{M}} / {M}} \operatorname{tr}(\CovYn)+\sigma^{2}.
\end{align}
\end{Lemma}


\subsection{Proof of \Cref{Prop:lap}} \label{app:PF}
We first show that ${\bm v}_1$ is a positive vector regardless of the choice of GSO. 
When $\GSO = {\bm L}$, $\GSO = \nL$, it is know that ${\bm v}_1 = {\bm 1} / \sqrt{N}$, ${\bm v}_1 = {\bm D}^{1/2}{\bf 1}/\|{\bm D}^{1/2}{\bf 1}\|_2 > {\bm 0}$, respectively; see \cite[p.~4-7]{spectrachung97}. 
When $\GSO = {\bm A}$, $\GSO = \nA$, we note that both matrices are non-negative. The Perron-Frobenius theorem implies ${\bm v}_1$ is a positive vector \cite[Theorem 8.4.4]{matranly_2012}. 

Note that in all cases, the graph frequency $\lambda_1$ which corresponds to ${\bm v}_1$ has a multiplicity of one. 
As such, to show that ${\bm v}_1$ is the \emph{only} positive eigenvector, it follows from the fact that every vector orthogonal to a positive vector must have at least one positive and negative element.\vspace{-.2cm}



\subsection{Proof of \Cref{the:kmeans}}
For ${\bm N} \in \RR^{N \times M}$, the $K$-means score \eqref{eq:kmeansUK} can be expressed via searching for $K$ centroid vectors $\bar{\bm n}_1,..., \bar{\bm n}_K \in \RR^M$:
\beq 
\Kmeans( {\bm N} ) = \min_{ \begin{subarray}{c}
\bar{\bm n}_1,..., \bar{\bm n}_K \in \RR^M 
\end{subarray} } \sum_{\ell=1}^N \min_{j=1,...,K} \| {\bm n}_\ell - \bar{\bm n}_j \|^2.
\eeq 
Define the set
\[ 
\mathscr{R}_K^{ N \times M}=\left\{ \overline{\bm N} \in R^{N \times M}: \text{ no. of unique rows of } \overline{\bm N} \leq K \right\}.
\]
It is easy to observe that 
\beq \label{eq:kmeans_obj} \textstyle
\Kmeans( {\bm N} ) = \min_{ \overline{\bm N} \in \mathscr{R}_K^{ N \times M} } \| {\bm N} - \overline{\bm N} \|_{\fro}^2.
\eeq 
To bound $\Kmeans( \V_K )$, we invoke Lemma~\ref{lem:Vpopu} and consider the population eigenvector matrix $\fancyV_K \fancyO_{K}$ defined therein. Notice that as $\fancyV_K \in \mathscr{R}_K^{ N \times K}$ [cf.~\eqref{eq:fancyVKform}] 
and $\fancyO_K$ is a $K \times K$ orthogonal matrix, we have  $\fancyV_K\fancyO_{K} \in \mathscr{R}_K^{ N \times K}$. By \Cref{lem:Vpopu}, the following holds with probability at least $1-2 / N$,
\beq
\begin{split} 
\Kmeans(\bm V_K) & = \min_{ \overline{\bm V} \in \mathscr{R}_K^{ N \times K} } \| \V_K - \overline{\bm V} \|_{\fro}^2 \leq \|{\bm V}_{K}-\fancyV_{K} \fancyO_{K} \|^2_{\fro} \\
& \leq 
(  p( N-K ) )^{-1} 
{35^2 K^3 \log N}.
\end{split} 
\eeq

\subsection{Proof of Theorem~\ref{theo:K1T}}\label{app:pftheoK1T0}
To simplify notation, we assume without loss of generality that 
$\| \hatsu_i - (\hatsu_i)_+ \|_1 \leq \| \hatsu_i + (-\hatsu_i)_+ \|_1$ for any $i=1,\ldots,N$.
Therefore, the positivity function is simplified to 
\[ 
\Pos( \hatsu_i ) = \| \hatsu_i - (\hatsu_i)_+ \|_1. 
\]
We proceed by bounding $\Pos (\hatsu_i )$ under different cases:\vspace{.1cm}

\noindent \underline{\emph{1) When $\Tgnd = {\cal T}_{0}$}}: We have $\barsu_1=\sv_1 > {\bm 0}$ and 
\begin{equation*} \label{eq:gamma1}
\begin{aligned}
\Pos\left(\hatsu_{1}\right) \leq \left\|\hatsu_{1}-\barsu_{1}\right\|_{1}+ {\left\|\barsu_{1}-\left(\barsu_{1}\right)_{+}\right\|_{1}}
= \left\|\hatsu_{1}-\barsu_{1}\right\|_{1}.
\end{aligned}
\end{equation*}
where we used $| \min\{a,0\} - \min\{b,0\} | \leq |a-b|$ for any $a,b \in \RR$. For any $j \geq 2$, similarly we have:
\begin{equation*} \label{eq:gammaj}
\begin{aligned}
\Pos\left(\hatsu_{j}\right) \geq { \|\barsu_{j} - \left(\barsu_j\right)_{+} \|_{1}} - \left\|\hatsu_{j}-\barsu_{j}\right\|_{1} 
\geq {\rm c}_0 - \left\|\hatsu_{j}-\barsu_{j}\right\|_{1} ,
\end{aligned}
\end{equation*}
where we have used the fact that for any $j \geq 2$, it holds $\barsu_j = \sv_{j'}$ for some $2 \leq j' \leq N$.\vspace{.1cm}

\noindent \underline{\emph{2) When $\Tgnd = {\cal  T}_{1}$}}: As $\HS$ is \emph{not $1$-low pass}, it holds
\begin{equation*} \label{eq:gammat11}
\begin{aligned}
\Pos\left(\hatsu_{1}\right) \geq { \left\|\barsu_{1}-\left(\barsu_{1}\right)_{+}\right\|_{1}} - \left\|\hatsu_{1}-\barsu_{1}\right\|_{1}
\geq {\rm c}_0 - \left\|\hatsu_{1}-\barsu_{1}\right\|_{1}
\end{aligned}
\end{equation*}
since $\barsu_1 = \sv_{j'}$ for some $2 \leq j' \leq N$. 
Meanwhile, there exists $j^\star \in \{2,\ldots,N\}$ such that $\barsu_{ j^\star } = \sv_1 > {\bm 0}$. We also bound
\begin{align}
& \min_{ j = 2,\ldots, N } \Pos\left(\hatsu_{j}\right) \leq 
\left\|\hatsu_{j^\star}-\left(\hatsu_{j^\star}\right)_{+}\right\|_{1} \\
& \leq \left\|\barsu_{j^\star}-\left(\barsu_{j^\star}\right)_{+}\right\|_{1} + \left\|\hatsu_{j^\star}-\barsu_{j^\star}\right\|_{1}
= \left\|\hatsu_{j^\star}-\barsu_{j^\star}\right\|_{1}. \notag
\end{align}

From the above discussions, a sufficient condition for the detector to be accurate, i.e., $\hatTnec = \Tgnd$, is
\beq \textstyle \label{eq:c0cond}
{\rm c}_0 \geq 2 \max_{j=1,\ldots,N} \| \hatsu_{j}-\barsu_{j} \|_{1}. 
\eeq 
Applying \Cref{lem:sample} with $K=1$ yields
\begin{equation}
\begin{aligned}
&\max_{j=1,\ldots, N} \|\barsu_j-\hatsu_j\|_1 \leq 2^{1.5}\sqrt{N}\Emc/\eigap^{\sf min},
\end{aligned} 
\end{equation}
where $\Emc$ is defined in \eqref{eq:cov_conc_K} and the extra $\sqrt{N}$ is due to norm equivalence. Rearranging terms conclude the proof.\vspace{-.2cm}

\subsection{Proof of Theorem~\ref{theo:Km_sample}} \label{app:pftheoKm_sam}
We proceed by bounding the sample complexity under different cases for $\Tgnd$. 
\vspace{.1cm}

\noindent \underline{\emph{1) When $\Tgnd = {\cal T}_{0}$}}: As $\HS$ is a $K$ low pass graph filter, we observe that $\U_K = \V_K \Pi_K$ for some permutation matrix $\Pi_K \in \{0,1\}^{K \times K}$. Define the orthogonal matrix $\barfanO_{K}= \fancyO_K \Pi_K \hatO_{K}$, where $\fancyO_K, \hatO_K$ are from Lemma~\ref{lem:Vpopu}, \ref{lem:sample}, respectively. It holds with probability at least $1-2/N-5/M$,
\beq \label{eq:kmeans_upperbd}
\begin{split}
& \Kmeans( \hatU_K ) \leq \|\hatU_{K}-\fancyV_{K} \barfanO_{K}\|^2_{\fro} \\
& \leq 2 \|\hatU_{K}-\barU_{K} \hatO_K\|^2_{\fro} + 2 \|\barU_{K} \hatO_K-\fancyV_{K} \barfanO_{K}\|^2_{\fro} \\
& \leq 16 K \left( \frac{\Emc}{ \bareig } \right)^2 + \frac{2450 K^3 \log N}{  p( N-K ) }.
\end{split}
\eeq

\noindent \underline{\emph{2) When $\Tgnd = {\cal T}_1$}}:  As $\HS$ is \emph{not $K$ low pass}, the columns of $\barU_K$ has at least one eigenvector from $\{ {\bm v}_{K+1}, \ldots, {\bm v}_N \}$. 

To facilitate our analysis, we define the shorthand notations $\hatU_{r,s}= \big[ \hatsu_r,\ldots,\hatsu_s \big]$, $\U_{r,s}= \big[ \su_r,\ldots,\su_s \big]$ where $r \leq s$.
We also define the permutation function $\pi : \{1,\ldots, N\} \to \{ 1, \ldots, N \}$ such that $|h_i| = |h( \lambda_{\pi(i)} )|$, where $\pi(i)$ is the index of graph frequency for the $i$th largest frequency response\footnote{Recall that $|h_1| > \cdots > |h_N|$ are  sorted in descending order.} in $\HS$ and is well-defined under H\ref{ass:dis}. 
Set
\beq \label{eq:eigset}
 \mathcal{P}^{\sf high} = \{ i \,:\, 1 \leq i \leq K, K+1 \leq \pi(i) \leq N \}
\eeq 
to be the set of \emph{crossed} frequencies which is non-empty under $\Tgnd = {\cal T}_1$.

Let ${\bm O}\in \mathbb{R}^{K \times K}$ be an orthogonal matrix. From \eqref{eq:kmeans_obj}, we obtain that $\Kmeans(\V_{K}{\bm O})=\min _{\M \in  \mathscr{R}_K^{N \times K}}\|\V_{K}-\M{\bm O}^\top\|_F^2$. Since $\M{\bm O}^\top \in \mathscr{R}_K^{N \times K}$, we have $\Kmeans(\V_{K}{\bm O})=\Kmeans(\V_{K})$. This indicates that $K$-means score is invariant to multiplication by orthogonal matrices.
For any $r \leq s$ with $[r,s] \subseteq {\cal P}^{\sf high}$ [cf.~\eqref{eq:eigset}],
\beq \label{eq:T1rs1}
\begin{aligned}
\sqrt{ \Kmeans(\hatU_K) } &= \sqrt{ \Kmeans(\hatU_K\hatO_{K}) }  \\
& \textstyle = \min_{ \M \in \mathscr{R}_K^{N \times K} } \| \hatU_{K}\hatO_{K} - \M \|_{\fro} \\ 
&\geq \sqrt{ \Kmeans(\U_{K}) }
- \| \U_{K}-\hatU_{K}\hatO_{K} \|_{\fro}\\
&\geq \sqrt{ \Kmeans(\U_{r,s}) }
- \| \U_{K}-\hatU_{K}\hatO_{K} \|_{\fro}. 
\end{aligned}
\eeq
By the definition of ${\cal P}^{\sf high}$, $\U_{r,s}$ consists of at least one vector from $\{ \sv_{K+1}, \ldots, \sv_N \}$. From \Cref{conj:eig}, we have $\Kmeans(\U_{r,s}) \geq {\rm c}_{\sf SBM}$.
It can be shown that \eqref{eq:sample_complexity_K} implies $\sqrt{{\rm c}_{\sf SBM}} > \frac{2^{3/2}\sqrt{K}}{ \bareig } \Emc$. Thus applying \Cref{lem:sample} shows that the following lower bound holds with probability at least $1-5/M$,
\beq \label{eq:kmeans_lowerbd}
\Kmeans(\hatU_K) \geq \left( \sqrt{{\rm c}_{\sf SBM}} - \frac{2^{3/2}\sqrt{K}}{ \bareig } \Emc\right)^2.
\eeq 

Collecting \eqref{eq:kmeans_upperbd}, \eqref{eq:kmeans_lowerbd} and observe that $\hatTsuff = \Tgnd$ is guaranteed if {\sf (i)} $\delta$ upper bounds the RHS of \eqref{eq:kmeans_upperbd}, and {\sf (ii)} $\delta$ lower bounds the RHS of \eqref{eq:kmeans_lowerbd}. The proof is concluded.

\subsection{Proof of Lemma~\ref{lem:Vpopu}} \label{app:vpop}
Let $\sin \Theta(\V_{K}, \fancyV_{K})$ be a diagonal matrix whose $i$th diagonal element is $\sin ( \cos^{-1} ( \sigma_i ) )$ and $\sigma_i$ is the $i$th singular value of $\V_{K}^\top \fancyV_K$. 
Notice that $\fancyLn={\bm I}-\fancyAn$ and $\nL={\bm I}-\nA$, we first apply the Davis-Kahan theorem \cite{2015YUdaviskahan} 
to bound the subspace difference between $\V_K, \fancyV_K$:
\begin{equation}
\begin{aligned}
&\| \sin \Theta(\V_{K}, \fancyV_{K}) \|_{\rm F} 
\leq 
\frac{ 2 \sqrt{K}\| \nA - \fancyAn  \|_{2} }{ \lambda_{K}^{\fancyAn} - \lambda_{K+1}^{\fancyAn}} ,
\end{aligned}
\end{equation}
where 
$\lambda_{j}^{\fancyAn}$ denotes the $j$th largest eigenvalue of $\fancyAn$. Theorem 4 of \cite{sbmyu2014impact} shows that when the minimum expected degree of nodes satisfies $d_{\min, N}\geq 32\log N$, with probability at least $1-2 / N$, it holds
\begin{equation}
\|\fancyAn - \nA \|_2 \leq 10 \sqrt{ \log(N) / d_{\min, N} }.
\end{equation}
We have $d_{\min , N}=\frac{N(p+r K)}{K}-(p+r)$. The condition $\frac{p}{K}+r \geq \frac{32\log N+1}{N}$ guarantees $d_{\min, N}\geq 32\log N$. By \cite[p.~34-37]{2011RoheSBM}, we have $\lambda_{K}^{\fancyAn}=\frac{ p }{ rK + p}$ and $\lambda_{K+1}^{\fancyAn}=0$. As such,
\beq
\begin{aligned}
&\| \sin \Theta(\V_{K}, \fancyV_{K}) \|_{\rm F} \\&
\leq \frac{20\sqrt{K}\sqrt{\log N}}{\sqrt{\frac{N(p+rK)}{K}-(p+r)} \, \frac{p}{p+rK}} 
\leq  \frac{25 K^{1.5}\sqrt{\log N}}{\sqrt{p}\sqrt{N-K}} ,
\end{aligned}
\eeq

Finally, we observe that with probability at least $1-2/N$,
\[ 
\|\V_{K}-\fancyV_{K}\fancyO_{K}\|_{\fro} \leq \sqrt{2} \|\sin \Theta(\V_{K}, \fancyV_{K})\|_{\fro}\leq
\frac{35K^{1.5}\sqrt{\log N}}{\sqrt{p}\sqrt{N-K}} .
\]

\bibliographystyle{IEEEtran}
\bibliography{ref}

\clearpage
\newpage
\setcounter{page}{1}

\title{Supplementary Material for ``Detecting Low Pass Graph Signals via Spectral Pattern: Sampling Complexity and Applications''}
\author{Chenyue Zhang, Yiran He, Hoi-To Wai\thanks{The authors are with the Department of SEEM, The Chinese University of Hong Kong, Shatin, Hong Kong SAR of China. E-mails: \url{czhang@se.cuhk.edu.hk}, \url{yrhe@se.cuhk.edu.hk}, \url{htwai@se.cuhk.edu.hk}. This work is supported in part by CUHK Direct Grant \#4055135 and HKRGC Project \#24203520.}}

{\let\newpage\relax\maketitle}
\appendix 
This supplementary document showcases the applicability of \Cref{algo:nec} for tackling \Cref{prob:det} with non-stationary graph signals, i.e., when ${\bm C}_x = \EE[ {\bm x}_m {\bm x}_m^\top ] \neq {\bm I}$.\vspace{.1cm}

\noindent \textbf{{\sf (I)} Observation Model.} 
We examine a general model of non-white excitation where $ {\bm x}_m = \B {\bm z}_m$. Without loss of generality, the latent parameter vector $ {\bm z}_m \in \RR^{ R}$ follows a zero-mean, sub-Gaussian distribution with the covariance $\EE[ {\bm z}_m {\bm z}_m^\top ] = {\bm I}$, and the matrix $\B \in \RR^{N \times R}$, $R \leq N$. Note that this implies the excitation covariance $\EE[ {\bm x}_m {\bm x}_m^\top ] = {\bm B} {\bm B}^\top$.
The observed signal in Eq.~\eqref{eq:y_sec} can be rewritten as:
\beq \label{eq:y_sec_B}
{\bm y}_m = \underbrace{\HS \cdots \HS}_{J_m~\text{times}}\B {\bm z}_m + {\bm w}_m,~m=1,...,M,
\eeq
with the noiseless signal covariance matrix:
\beq \label{eq:cov_y_nl_B} \textstyle 
\begin{aligned}
\CovY & \textstyle = \frac{1}{J}\sum_{\tau=1}^J [\HS]^\tau \B \B^\top [\HS]^\tau \\
&= \V \left( \frac{1}{J} \sum_{\tau=1}^J h( \bm{\Lambda} )^{ \tau} 
\V^\top \B \B^\top \V h( \bm{\Lambda} )^{ \tau}\right) \V^\top .
\end{aligned}
\eeq

We now illustrate that \eqref{eq:cov_y_nl_B} admits the spectral pattern described in Section~\ref{sec:det} in an approximate fashion provided that the graph filters have sufficiently `sharp' frequency responses (to be discussed later) and $K \leq R$. Note that this enables us to tackle \Cref{prob:det} using \Cref{algo:nec}. 
To keep the discussion simple, we focus on the case where $J=1$. The noiseless covariance matrix can be written as:
\beq \label{eq:covy_def_B}
\CovY =\V \left(  h( \bm{\Lambda} )
\V^\top \B \B^\top \V h( \bm{\Lambda} )\right) \V^\top=\barUU \barH \barUU^\top.
\eeq
where $\barH = {\rm Diag}( \beta_1(\CovY), \ldots, \beta_N( \CovY ) )$ is a diagonal matrix for the eigenvalues of $\CovY$ sorted in descending order, and $\barUU$ is the eigenvector matrix of $\CovY$. We consider two hypothesis for $\HS$ as follows. \vspace{.1cm}

\noindent \underline{\emph{1) When $\Tgnd = {\cal T}_{0}$}}: Suppose that the graph filter is sufficiently `sharp', i.e.,
\beq \label{eq:lowpass_B}
\eta_K = \frac{ \max \{ |h(\lambda_{K+1})|, \ldots, |h(\lambda_N)| \} }{ \min\{ |h(\lambda_1)|, \ldots, |h(\lambda_K)| \}  } \ll 1.
\eeq 
In this case, we have the approximation 
\[
\CovY \approx \V_K \left(  h( \bm{\Lambda_K} )
\V_K^\top \B \B^\top \V_K h( \bm{\Lambda_K} )\right) \V_K^\top.
\]
Consequently, it can be deduced that the top-K eigenvectors, $\hatU_K$, of the observation covariance matrix satisfies ${\rm span} (\hatU_K) \approx {\rm span} (\V_K)$. It follows that both $\Kmeans( \hatU_K ), \Pos ( \hatsu_1 )$ will be small when $\Tgnd = {\cal T}_0$.\vspace{.1cm}


\noindent \underline{\emph{2) When $\Tgnd = {\cal T}_{1}$}}: Similar to the previous case, we require the graph filter to satisfy the following sharpness condition:
\beq \label{eq:highpass_B}
\frac{|h_{K+1}|}{|h_K|} = \frac{|h( \lambda_{\pi(K+1)}|}{|h( \lambda_{\pi(K)}|} \ll 1,
\eeq 
where we recall the definition of $\pi(\cdot)$ such that $|h_i| = |h(\lambda_{ \pi(i) })|$ and $|h_i|$ is the $i$th highest frequency response of the graph filter. Similarly, we observe that $\CovY$ is approximated by 
\[
\V_{\pi([K])} \left( h( \bm{\Lambda_{\pi([K])}} )
\V_{\pi([K])}^\top \B \B^\top \V_{\pi([K])} h( \bm{\Lambda_{\pi([K])}} ) \right) \V_{\pi([K])}^\top
\]
Consequently, the set of top-K eigenvectors, $\hatU_K$, of the observed covariance matrix satisfies ${\rm span}( \hatU_K ) \approx {\rm span} ( \V_{\pi([K])})$. Since $\V_{\pi([K])}$ contains an eigenvector $\sv_i$ with $i \in {\cal P}^{\sf high}$ \eqref{eq:eigset}, it follows that $\Kmeans( \hatU_K ), \Pos( \hatsu_1 )$ are likely to be bounded away from zero when $\Tgnd = {\cal T}_1$.

The above derivation shows that if the admissible graph filters are sufficiently sharp, i.e., satisfying \eqref{eq:lowpass_B}, \eqref{eq:highpass_B}, then \Cref{algo:nec} is applicable to non-stationary graph signals with general excitation \eqref{eq:y_sec_B}. Note that this differs from the analysis for stationary graph signals in Section~\ref{sec:det}, \ref{sec:ana}, where the sharpness of graph filters only affects the sample complexity of the algorithm. To give further insight, we conclude this supplementary document by quantifying the effects of the sharpness of graph filters under finite number of samples. 
\vspace{.1cm}


\noindent \textbf{{\sf (II)} Performance Analysis when \texorpdfstring{$\mathbf{B}\neq \mathbf{I}$}{}.} For simplicity, we further concentrate on applying \Cref{algo:nec} with $K \geq 2$. Our techniques can be readily extended to other cases such as when $J\geq 2$ or $K=1$.
To fix ideas, the eigenvalues of $\HS$ are arranged in descending order as $|h_1| > \cdots > |h_N|$. The \emph{strength} (a.k.a.~sharpness) of the graph filter at the $i$th graph frequency is defined as $\bar{\eta}_i := |h_{i+1}|/|h_{i}| $. When $\HS$ is $K$-low pass, we observe that $\bar{\eta}_K$ coincides with the low pass ratio $\bar{\eta}_K = \eta_K < 1$. 

We consider a mild condition for $\B$ and $\HS$. Let ${\bm Q}_K$ be the top-$K$ right singular vectors of $\HS \B$. We require:
\begin{assumption} \label{ass:skeB}
Assume: {\sf (i)} $\operatorname{rank}(\HS \B) \geq K$, {\sf (ii)} $ \beta_K > \beta_{K+1}$, and 
{\sf (iii)} $\operatorname{rank} ( \barU_{K}^{\top} \B \Q_{K} )=K$.
\end{assumption}
\noindent We recall that $\beta_i := \beta_i( \CovY)$ as defined in \eqref{eq:covy_def_B}. The above conditions are relatively mild, e.g., they typically hold when ${\bm B}$ has full column rank. 

Note that the eigenvectors of $\CovY$ are now \emph{perturbed} version for the columns of $\{\sv_1,\ldots,\sv_N\}$. Recall that $\barU_K$ is the top-K eigenvectors of the noiseless covariance $\CovY$ of the \emph{stationary graph signals} in \eqref{eq:y_sec}, our first step is to control the difference between $\barU_K$ and $\barUU_{K}$:
\begin{Lemma}\label{lem:skechinf}
Fix $1 \leq K \leq R$, and assume \Cref{ass:skeB}. There exists an orthogonal matrix $\barO_K $ such that 
\beq
 \| \barUU_{K}  - \barU_{K} \barO_{K}^\top\|_{\fro}^2 \leq \bar{\eta}^{2}_K \, {2K} \ff.
\eeq
where $\ff = \|\barU_{N-K}^{\top} \B {\bm Q}_{K} \|_{2}^2 \| ( \barU_{K}^{\top} \B {\bm Q}_{K} )^{-1} \|_{2}^2$.  
\end{Lemma} 
\begin{proof}
 Applying \cite[Proposition 1]{2018ToBlind}, we obtain
\beq \label{eq:applem4to}
\| \U_{K} (\U_{K})^\top - \barUU_{K} \barUU_{K}^\top \|_{2}^2 \leq \bar{\eta}_{K}^2 \ff.
\eeq    
From \cite[Lemma 7]{Gittens2013ApproximateSC}, we then have $\|\barUU_{K}-\barU_{K} \barO_{K} \|_{\mathrm{F}}\leq \sqrt{2K} \|\barU_K \barU_K^\top -\barUU_K \barUU_K^\top \|_{2}$.
This concludes the proof.
\end{proof}
\noindent Note that in the above, we have $\ff = 0$ if $\B = \I$. The above lemma leads to the following theorem on the finite sample performance of \Cref{algo:nec}:
\begin{Theorem}\label{theo:Km_sample_B}
Let $G \sim {\sf SBM}(K, N, r, p)$ with $p \geq r > 0$, $\frac{p}{K}+r \geq \frac{32\log N+1}{N}$, \Cref{ass:skeB}  holds, and take $\GSO = \nL$ as the unweighted adjacency matrix. 
Accordingly, let 
\beq \notag
\begin{aligned}
\bareig & := \inf_{ \eigB_{K} : \HS \in {\cal H}_{\GSO}^{\sf low} \cup {\cal H}_{\GSO}^{\sf hi} } \hspace{-.2cm} \eigB_{K},~\bar{\eta} := \sup_{ \bar{\eta}_{K} :  \HS \in {\cal H}_{\GSO}^{\sf low} \cup {\cal H}_{\GSO}^{\sf hi} } \hspace{-.2cm} \bar{\eta}_K, \\
\delta_1 & \textstyle :=\sqrt{\frac{\delta}{2}- \frac{1225 K^{3} \log N}{p (N-{K})}-2K\bar{\eta}^2\ff},\\
\delta_2 & := \sqrt{{\rm c}_{\sf sbm}} -\sqrt{2K\ff}\bar{\eta}-\sqrt{\delta},
\end{aligned}
\eeq 
Assume that $\bareig > 0$, $\overline{\delta}_{\sf min} := \min \left\{\delta_1,\delta_2\right\}>0$,
and the noise variance $\sigma^{2}\leq \overline{\delta}_{\sf min} {\bareig} / {\sqrt{8K}}$. If it holds that
\beq \label{eq:sample_complexity_K_B}
\sqrt{\frac{M}{\log{M}}} \geq  
\frac{2 {\rm c}_{1} \operatorname{tr}(\CovYn)}{\overline{\delta}_{\sf min}\bareig / \sqrt{8K}-\sigma^{2} },
\eeq
then 
\beq 
{\rm Pr}( \hatTsuff = \Tgnd ) \geq 1 - 10/M - 2/N.
\eeq 
\end{Theorem}
\begin{proof}
We proceed by modifying the proof of \Cref{theo:Km_sample} in \cref{app:pftheoKm_sam}.\vspace{.1cm}

\noindent \underline{\emph{1) When $\Tgnd = {\cal T}_{0}$}}:  Define the orthogonal matrices $\tilde{\bm O}_{K}= \barO_K \hatO_{K}$, $\tilde{\fancyO}_{K}= \fancyO_K \Pi_K\tilde{\bm O}_{K}$, where $\fancyO_K,\hatO_K, \barO_K$ are from Lemma~\ref{lem:Vpopu}, \ref{lem:sample}, \ref{lem:skechinf}, respectively. Combining \eqref{eq:kmeans_upperbd} with \cref{lem:skechinf}, it holds with probability at least $1-2/N-5/M$ that
\beq \label{eq:kmeans_upperbd_B}
\begin{split}
& \Kmeans( \hatU_K ) \leq \|\hatU_{K}-\fancyV_{K} \tilde{\fancyO}_{K}\|^2_{\fro} \\
& \leq 2 \|\hatU_{K}-\barUU_{K} \hatO_K\|^2_{\fro} + 2 \|\barUU_{K} \hatO_K-\barU_{K} \tilde{\bm O}_{K}\|^2_{\fro}\\
&+2 \|\barU_{K} \tilde{\bm O}_{K}-\fancyV_{K} \tilde{\fancyO}_{K}\|^2_{\fro} \\
& \leq 16 K \left( \frac{\Emc}{ \bareig } \right)^2+
4K\bar{\eta}^2\ff
+ \frac{2450 K^3 \log N}{  p( N-K ) }.
\end{split}
\eeq 

\noindent \underline{\emph{2) When $\Tgnd = {\cal T}_1$}}: Combining \eqref{eq:T1rs1} with \Cref{lem:skechinf}, we have
\begin{equation*}
\begin{aligned}
&\sqrt{ \Kmeans(\hatU_K) } = \sqrt{ \Kmeans(\hatU_K\hatO_{K}) } \\
& \geq \sqrt{ \Kmeans(\barUU_K) }- \| \barUU_K-\hatU_K\hatO_{K} \|_{\fro} \\
& \geq \sqrt{ \Kmeans(\barU_K\barO_{K}) }- \| \barU_K\barO_{K}-\barUU_K \|_{\fro}-\| \barUU_K-\hatU_K\hatO_{K} \|_{\fro} \\
& \geq \sqrt{ \Kmeans(\barU_{r,s}) }- \| \barU_K\barO_{K}-\barUU_K \|_{\fro}-\| \barUU_K-\hatU_K\hatO_{K} \|_{\fro}\\
\end{aligned}
\end{equation*}
As it can be shown from \eqref{eq:sample_complexity_K_B} that $\sqrt{{\rm c}_{\sf SBM}} > \frac{2^{3/2}\sqrt{K}}{ \bareig } \Emc+\bar{\eta}_{K}\sqrt{2K\ff}$, we have the following lower bound that holds with probability at least $1-5/M$,
\beq \label{eq:kmeans_lowerbd_B}
\Kmeans(\hatU_K) \geq \left( \sqrt{{\rm c}_{\sf SBM}} - \frac{2^{3/2}\sqrt{K}}{ \bareig } \Emc-
\bar{\eta} \sqrt{2K\ff}
\right)^2.
\eeq 
Collecting \eqref{eq:kmeans_upperbd_B}, \eqref{eq:kmeans_lowerbd_B} concludes the proof.
\end{proof}

The above result highlights a salient difference between the case of non-stationary graph signals \eqref{eq:y_sec_B} and stationary graph signals \eqref{eq:y_sec} considered in the main paper. Particularly, satisfying the assumptions in \Cref{theo:Km_sample_B} necessitates $\sqrt{ {\rm c}_{\sf SBM} } > \sqrt{2K \ff} \bar{\eta}$. That is, the sharpness of the graph filters in consideration has to be higher than a certain threshold specified by the non-clusterability of the bulk eigenvectors, \emph{regardless} of the number of sample observed and the noise variance. This restriction is not found in the case of stationary graph signals \Cref{theo:Km_sample}, and it can be understood from the intuition that follows \eqref{eq:lowpass_B}, \eqref{eq:highpass_B}. In particular, when $\bar{\eta} \approx 1$, the approximation of the top-K eigenvectors in $\CovY$ as $\V_{ \pi( [K] ) }$ is no longer valid and the top-K eigenvectors in $\CovY$ suffers from an \emph{intrinsic perturbation} from $\V_{ \pi( [K] ) }$ due to the effects of the matrix ${\bm B}$.

\end{document}